\documentclass[a4paper]{amsart}
\pdfoutput=1

\usepackage[utf8]{inputenc}
\usepackage{lmodern}
\usepackage{microtype}

\usepackage[english]{babel}
\usepackage{csquotes} % removes one warning of biblatex

\usepackage{setspace}
\usepackage{amsmath,amsthm,amsfonts,amssymb}
\usepackage{mathtools}
\usepackage{accents}    % Pour underaccent
\usepackage[only,llbracket,rrbracket]{stmaryrd} % pour llbracket et rrbracket

\usepackage[shortlabels]{enumitem}

\usepackage{xcolor}

\usepackage[pdftex,hypertexnames=false]{hyperref}

\definecolor{darkblue}{rgb}{0,0,0.5}
\definecolor{cerule}{RGB}{53,122,183}
\definecolor{cardinal}{RGB}{184,32,16}

\hypersetup{
    colorlinks,
    linkcolor={black},
    citecolor={cerule},
    urlcolor={cardinal}
}

\usepackage{float}
\usepackage[noend]{algpseudocode}
\usepackage[ruled]{caption}      % La légende des boites flottantes
\usepackage{subfig}

\usepackage{tikz}
\usetikzlibrary{arrows}

\usepackage{xifthen}

\floatstyle{ruled}
\newfloat{algo}{tp}{lop}
\floatname{algo}{Algorithm}
\usepackage[
  citestyle=authoryear-comp,
  bibstyle=authoryear,
  backend=bibtex,
  isbn=false,
  maxcitenames=2,
  maxbibnames=6,
  uniquename=init,
  firstinits=true,
  sortcites=false
]{biblatex}
\bibliography{biblio}

\DeclareFieldFormat{eprint:tel}{%
  \textsc{Tel}\addcolon\space\ifhyperref
    {\href{https://hal.archives-ouvertes.fr/tel-#1}{\nolinkurl{#1}}}
    {\nolinkurl{#1}}}

\newcommand\mop[1]{\mathop{\operatorname{#1}}}
\def\Res{\mop{Res}}
\def\diag{\mop{diag}}
\def\res{\mop{res}}
\def\lm{\mop{lm}}

\def\A{\mathbb{A}}
\def\C{\mathbb{C}}
\def\N{\mathbb{N}}
\def\Q{\mathbb{Q}}
\def\Z{\mathbb{Z}}
\def\K{{\mathbb{K}}}
\def\R{\mathbb{R}}

\def\F{\mathbb{F}}
\def\L{\mathbb{L}}

\def\smile{\smallsmile}
\def\leq{\leqslant}
\def\geq{\geqslant}

\def\cB{\mathcal{B}}
\def\cD{\mathcal{D}}
\def\cS{\mathcal{S}}
\def\cL{\mathcal{L}}
\def\cM{\mathcal{M}}

\def\allbigallsmall{\textsc{AllLargeOrAllSmall}}

\newcommand{\abs}[1]{\left| #1 \right|}
\newcommand{\slf}[1]{(\!(#1)\!)}

\def\ud{\mathrm{d}}
\def\eqdef{\stackrel{\text{def}}{=}}
\newcommand{\st}{\mathrel{}\middle|\mathrel{}}

\newcommand{\ssum}{\mathop{\sum\nolimits^{\mathrlap{\prime}}}}
\newcommand{\ul}{\underaccent{\bar}}

\newtheorem{theorem}{Theorem}[section]
\newtheorem{proposition}[theorem]{Proposition}
\newtheorem{lemma}[theorem]{Lemma}
\newtheorem{corollary}[theorem]{Corollary}
\newtheorem{conjecture}[theorem]{Conjecture}

\theoremstyle{definition}
\newtheorem{definition}[theorem]{Definition}
\newtheorem{example}{Example}

\theoremstyle{remark}

\title{Multiple binomial sums}
\author[A. Bostan]{Alin Bostan}
\address[A. Bostan]{
  Inria Saclay Île-de-France\\ 
  Bâtiment Alan Turing\\ 
  1 rue Honoré d'Estienne d'Orves\\ 
  91120 Palaiseau\\
  France
}
\email{
  alin.bostan@inria.fr
}

\author[P. Lairez]{Pierre Lairez}
\address[P. Lairez]{
  Technische Universität Berlin\\
  Fakultät II, Sekretariat 3-2\\
  Straße des 17. Juni 136\\
  10623 Berlin\\
  Germany
}
\email{
  pierre@lairez.fr
}

\author[B. Salvy]{Bruno Salvy}
\address[B. Salvy]{
  Inria, LIP (U. Lyon, CNRS, ENS Lyon, UCBL)\\
  France}
\email{
  bruno.salvy@inria.fr
}

\date{\today}
\subjclass[2010]{
  05A10 % Combinatorics // Enumerative combinatorics // Factorials, binomial coefficients, combinatorial functions
  (33F10 % Special functions // Computational aspects // Symbolic computation (Gosper and Zeilberger algorithms, etc.)
  68W30)% Computer science // Algorithms // Symbolic computation and algebraic computation
}

\keywords{Binomial sum, multiple sum, symbolic computation, diagonal, integral representation}

\begin{document}

\begin{abstract}
  Multiple binomial sums form a large class of multi-indexed sequences, closed under partial summation, which contains most of the sequences obtained by multiple summation of products of binomial coefficients and also all the sequences with algebraic generating function.  We study the representation of the generating functions of binomial sums by integrals of rational functions.  The outcome is twofold. Firstly, we show that a univariate sequence is a multiple binomial sum if and only if its generating function is the diagonal of a rational function.  Secondly, we propose algorithms that decide the equality of multiple binomial sums and that compute recurrence relations for them.  In conjunction with geometric simplifications of the integral representations, this approach behaves well in practice.  The process avoids the computation of certificates and the problem of the appearance of spurious singularities that afflicts discrete creative telescoping, both in theory and in practice.
\end{abstract}

\maketitle

\section*{Introduction}

The computation of definite sums in computer algebra is classically handled by
the method of \emph{creative telescoping} initiated in the 1990s
by Zeilberger~\parencite{Zeilberger1990,Zeilberger1991,WilfZeilberger1992a}.
For example, it applies to sums like
\begin{equation}
  \sum_{k=0}^n \frac{4^k}{\binom{2k}{k}} \text{,}\quad \sum_{k=0}^n\left(\sum_{j=0}^k \binom{n}{j}\right)^3 \text{ or }
  \sum_{i=0}^n \sum_{j=0}^n\binom{i+j}{j}^2\binom{4n-2i-2j}{2n-2i}.
  \label{eqn:sum-exem-intro}
\end{equation}
In
order to compute a sum~$\sum_k u(n,k)$ of a bivariate sequence~$u$, this method
computes an identity of the form
\[a_p(n)u(n+p,k)+\dotsb+a_0(n)u(n,k)=v(n,k+1)-v(n,k).\]
Provided that it is possible to sum both sides over~$k$ and that the
sequence~$v$ vanishes at the endpoints of the domain of summation, the
left-hand side---called a \emph{telescoper}---gives a recurrence for the
sum. The sequence~$v$ is then called the \emph{certificate} of the
identity.

In the case of multiple sums, this idea leads to searching for a telescoping identity of the form
\begin{multline}\label{eqn:dtelesc}
a_p(n)u(n+p,k_1,\dots,k_m)+\dots+a_0(n)u(n,k_1,\dots,k_m)=\\
\big(v_1(n,k_1+1,k_2,\dots,k_m)-v_1(n,k_1,\dots,k_m)\big)+\dotsb\\
+\big(v_m(n,k_1,\dots,k_m+1)-v_m(n,k_1,\dots,k_m)\big).
\end{multline}
Again, under favorable circumstances the sums of the sequences on the
right-hand side telescope, leaving a recurrence for the sum on the left-hand
side.

This high-level presentation hides practical difficulties.
It is important to check that the sequences on both sides of the
identities above are defined over the whole range of
summation~\parencite{Abramov2006,AbramovPetkovsek2005}.
More often than not, singularities do appear.  
To the best of our knowledge, no algorithm based on
creative telescoping manages to work around this difficulty; they all let the
user handle it. 
% This is also a weak point in the theory. In order to prove that
% every element of some class of sequences obtained by summation satisfies a
% linear recurrence relation, it is not sufficient to show that a telescopic relation
% such as~\eqref{eqn:dtelesc} exists: one also has to describe \emph{a priori}
% the singularities that may appear and to give a way for performing summations using a case-by-case analysis.  The second difficulty is a consequence of the first one: 
As a consequence, computing
the certificate is not merely a useful by-product of the algorithm, but indeed
a necessary part of the computation.  Unfortunately, the size of the
certificate may be much larger than that of the final recurrence and thus
costly in terms of computational complexity.

The computation of multiple integrals of rational functions has some
similarities with the computation of discrete sums and the method of creative
telescoping applies there too. It may also produce extra singularities in the
certificate, but in the differential setting this is not an issue anymore: for
the integrals we are interested in, the integration path can always be moved
to get around any extra singularity. Moreover, we have
showed~\parencite{BostanLairezSalvy2013,Lai15} that integration of multivariate
rational functions over cycles can be achieved efficiently without computing
the corresponding certificate and without introducing spurious singularities.
In that case, the algorithm computes a linear differential equation for the
parameterized integral. It turns out that numerous multiple sums can be cast
into problems of rational integration by passing to generating functions. The
algorithmic consequences of this observation form the object of the present
work.

\subsection*{Content}

In~\S\ref{sec:algbinomsums}, we define a class of multivariate sequences, called
\emph{(multiple) binomial sums}, that contains the binomial coefficient sequence and
that is closed under pointwise addition, pointwise multiplication, linear
change of variables and partial summation.  Not
every sum that creative telescoping can handle is a binomial sum: for example,
among the three sums in Eq.~\eqref{eqn:sum-exem-intro}, the second one and the third one
are binomial sums but the first one is not, since it contains the inverse of a
binomial coefficient; moreover, it cannot be rewritten as a binomial sum (see \S\ref{sec:caracbs}). Yet many sums coming from combinatorics and number theory
are binomial sums.  In~\S\ref{sec:gfun}, we explain how to compute \emph{integral
representations} of the generating function of a binomial sum in an automated way.
The outcome is twofold. Firstly, in \S\ref{sec:diagonals}, we work further on
these integral representations to show that the generating functions of
\emph{univariate} binomial sums are exactly the diagonals of rational power series.
This equivalence characterizes binomial sums in an
intrinsic way which dismisses the arbitrariness of the definition.  All the
theory of diagonals transfers to univariate binomial sums and gives many interesting
arithmetic properties.  Secondly, in \S\ref{sec:computing}, we show how to use integral
representations to actually compute with binomial sums (e.g. find recurrence
relations or prove identities automatically) \emph{via} the computation of
Picard-Fuchs equations.  The direct approach leads to
integral representations that involve far too many variables to be
efficiently handled.  In~\S\ref{sec:geomred}, we propose a general method, that we call \emph{geometric
reduction}, to reduce tremendously the number of variables in practice.
In~\S\ref{sec:opt}, we describe some variants with the purpose of
implementing the algorithms, and finally, in~\S\ref{sec:applications}, we show how the
method applies to some classical identities and more recent ones that were
conjectural so far.

All the algorithms that are presented here are implemented in Maple and are available at
  \url{https://github.com/lairez/binomsums}.

\subsection*{Example}
%\begin{example}\label{sec:exem-dixon}
The following proof of an identity of \textcite{Dix91}, 
\begin{equation}
  \sum_{k=0}^{2n} (-1)^k \binom{2n}{k}^3 = (-1)^n \frac{(3n)!}{n!^3},
  \label{eqn:intro:dixon}
\end{equation}
illustrates well the main points of the method.
The strategy is as follows: find an integral representation of the generating
function of the left-hand side; simplify this integral representation using
partial integration; use the simplified integral representation to compute a differential
equation of which the generating function is solution; transform this equation
into a recurrence relation; solve this recurrence relation.

First of all,
the binomial coefficient~$\binom{n}{k}$ is the coefficient of~$x^k$ in~$(1+x)^n$.
Cauchy's integral formula ensures that
\[ \binom{n}{k} = \frac{1}{2\pi i} \oint_{\gamma} \frac{(1+x)^n}{x^{k}} \frac{\ud x}{x}, \]
where~$\gamma$ is the circle~$\left\{ x\in\C \st |x|=\frac12 \right\}$.
Therefore, the cube of a binomial coefficient can be represented as a triple integral
\[ \binom{2n}{k}^3 = \frac{1}{(2\pi i)^3}\oint_{\gamma\times\gamma\times\gamma} \frac{(1+x_1)^{2n}}{x_1^{k}}\frac{(1+x_2)^{2n}}{x_2^{k}}\frac{(1+x_3)^{2n}}{x_3^{k}} \frac{\ud x_1}{x_1}\frac{\ud x_2}{x_2}\frac{\ud x_3}{x_3}. \]
As a result, the generating function of the left-hand side of Equation~\eqref{eqn:intro:dixon}
is
\begin{align*}
  y(t) &\eqdef \sum_{n = 0}^{\infty} t^n\sum_{k=0}^{2n} (-1)^k \binom{2n}{k}^3\\
  &=  \frac{1}{(2i\pi)^3}\oint_{\gamma^3}\sum_{n = 0}^{\infty} \sum_{k=0}^{2n} \left(t \prod_{i=1}^3(1+x_i)^2\right)^n \left(\frac{-1}{x_1x_2x_3}\right)^k  \frac{\ud x_1}{x_1}\frac{\ud x_2}{x_2}\frac{\ud x_3}{x_3}\\
  &= \frac{1}{(2i\pi)^3}\oint_{\gamma^3}\sum_{n = 0}^{\infty} \left(t \prod_{i=1}^3(1+x_i)^2\right)^n \frac{1 - \left(\frac{-1}{x_1x_2x_3}\right)^{2n+1}}{1+\frac{1}{x_1x_2x_3}}  \frac{\ud x_1}{x_1}\frac{\ud x_2}{x_2}\frac{\ud x_3}{x_3}\\
  &= \frac{1}{(2i\pi)^3}\oint_{\gamma^3}
  \frac{\left(x_1x_2x_3-t\prod_{i=1}^3(1+x_i)^2\right)\ud x_1\ud x_2\ud x_3}{\left(x_1^2x_2^2x_3^2-t\prod_{i=1}^3(1+x_i)^2\right)\left(1-t\prod_{i=1}^3(1+x_i)^2\right)}.
\end{align*}
The partial integral with respect to~$x_3$ along the circle~$|x_3|=\frac12$ is
the sum of the residues of the rational function being integrated at the poles
whose modulus is less than~$\frac12$.  When~$|t|$ is small
and~$|x_1|=|x_2|=\frac12$, the poles coming from the
factor~$x_1^2x_2^2x_3^2-t\prod_{i=1}^3(1+x_i)^2$ all have a modulus that is
smaller than~$\frac12$: they are asymptotically proportional to~$|t|^{1/2}$.  In contrast, the poles coming
from the factor~$1-t\prod_{i=1}^3(1+x_i)^2$ behave like~$|t|^{-1/2}$ and have
all a modulus that is bigger than~$\frac12$.  In particular, any two poles that
come from the same factor are either both asymptotically small or both
asymptotically large.  This implies that the partial integral is a rational
function of~$t$, $x_1$ and~$x_2$; and we compute that
\[ y(t) = \frac{1}{(2i\pi)^2}\oint_{\gamma\times\gamma} 
  \frac{ x_1x_2 \ud x_1\ud x_2}{x_1^2x_2^2-t (1+x_1)^2(1+x_2)^2(1-x_1x_2)^2}. \]
This formula echoes the original proof of \textcite{Dix91} in which the left-hand side of~\eqref{eqn:intro:dixon}
is expressed as the coefficient of~$(xy)^{4n}$ in~$((1-y^2)(1-z^2)(1-y^2z^2))^{2n}$.
Using any algorithm that performs definite integration of rational functions~\parencite{Chy00,Kou10,Lai15}
reveals a differential equation satisfied by~$y(t)$:
\[ t(27t+1)y''+(54t+1)y'+6y = 0. \]
Looking at the coefficient of~$t^n$ in this equality leads to
the recurrence relation
\[ 3(3n+2)(3n+1)u_n+(n+1)^2 u_{n+1} = 0, \]
where~$u_n = \sum_{k=0}^{2n} (-1)^k \binom{2n}{k}^3$.
Since~$u_0=1$, it leads to a proof of Dixon's identity by induction on~$n$.  The
treatment above differs in one important aspect from what follows: the use of genuine
integrals and explicit integration paths rather than \emph{formal residues} that
will be introduced in~\S\ref{sec:gfun}.
%\end{example}

\subsection*{Comparison with creative telescoping} 
As mentioned above, the computation of multiple binomial sums can be handled by
the method of creative telescoping.  The amount of work in this direction is
considerable and we refer the reader to surveys~\parencite{Chy14,Kou13a}. In the 
specific context of multiple sums, the most relevant works are those of 
\textcite{Wegschaider1997}, \textcite{Chy00},
\textcite{ApZe06},
and \textcite{GaroufalidisSun2010}.
We show on the example of Dixon's identity how the method of creative
telescoping and the method of generating functions differ fundamentally even on a single sum. 

Let~$u_{n,k} = (-1)^k\binom{2n}{k}^3$.
This bivariate sequence satisfies the recurrence relations
\begin{equation}
\begin{aligned}
  (2n+2-k)^3(2n+1-k)^3 u_{n+1,k} - 8(1+n)^3(1+2n)^3 u_{n,k} &= 0 \\
   \text{and}\quad (k+1)^3 u_{n,k+1} + (2n-k)^3u_{n,k} &= 0.
\end{aligned}
  \label{equ:recbinom}
\end{equation}
With these relations as input, Zeilberger's algorithm finds the sequence
\begin{multline*}
  v_{n,k} = \frac{P(n,k)}{2(2n+2-k)^3(2n+1-k)^3} u_{n,k}, \\
  \text{where}\quad P(n,k) = k^3(9k^4n-90k^3n^2+348k^2n^3-624kn^4+448n^5\\ +6k^4-132k^3n
  +792k^2n^2-1932kn^3+1760n^4-48k^3+594k^2n\\-2214kn^2
  +2728n^3+147k^2-1113kn+2084n^2-207k+784n+116),
\end{multline*}
that satisfies
\begin{equation}
  3(3n+2)(3n+1)u_{n,k}+(n+1)^2 u_{n+1,k} = v_{n,k+1} - v_{n,k}.
  \label{equ:telescopic-dixon}
\end{equation}
Whatever the way the sequence~$v_{n,k}$ is found, it is easy to check the
telescopic relation \eqref{equ:telescopic-dixon}: using the recurrence relations
for~$u_{n,k}$, each of the four terms in~\eqref{equ:telescopic-dixon}
rewrites in the form~$R(n,k)u_{n,k}$, for some rational function~$R(n,k)$.
However, for some specific values of~$n$ and~$k$, the sequence~$v_{n,k}$ is
not defined, due to the denominator.

To deduce a recurrence relation for~$a_n = \sum_{k=0}^{2n} u_{n,k}$, it is desirable to sum the telescopic relation~\eqref{equ:telescopic-dixon}, over~$k$,
from~$0$ to~$2n+2$.  Unfortunately, that would hit the forbidden set where~$v_{n,k}$ is
not defined.  We can only safely sum up to~$k=2n-1$. Doing so, we obtain
that
\begin{equation*}
  3(3n+2)(3n+1) \sum_{k=0}^{2n-1} u_{n,k} +(n+1)^2 \sum_{k=0}^{2n-1} u_{n+1,k} = v_{n,2n} - v_{n,0},
  %= n^3(8n^5+52n^4+146n^3+223n^2+185n+58).  
\end{equation*}
and then
\begin{multline*}
 3(3n+2)(3n+1)(a_n - u_{n,2n}) +(n+1)^2 (a_{n+1} - u_{n+1,2n+2}-u_{n+1,2n+1}-u_{n+1,2n}) \\
  = n^3(8n^5+52n^4+146n^3+223n^2+185n+58).  
\end{multline*}
It turns out that~$3(3n+2)(3n+1)u_{n,2n} +(n+1)^2 (
u_{n+1,2n+2}+u_{n+1,2n+1}+u_{n+1,2n})$ evaluates exactly to the right-hand side
of the above identity, and this leads to Dixon's identity.

In this example, spurious singularities clearly appear in the range of
summation.  Thus, deriving an identity such as Dixon's  from a telescopic
identity such as~\eqref{equ:telescopic-dixon} is not straightforward and
involves the certificate. 
A few works address this issue for single
sums \parencite{AbramovPetkovsek2005,Abramov2006}, but none for the case of multiple
sums: existing algorithms \parencite{Wegschaider1997,Chy00,GaroufalidisSun2010} only
give the telescopic identity without performing the summation.  A recent
attempt by \textcite{ChyMahSibTas14} to check the recurrence satisfied by Apéry's
sequence in the proof assistant Coq has shown how difficult it is to formalize
this summation step. Note that because of this issue, even the \emph{existence}
of a linear recurrence for such sums can hardly be inferred from the fact that the
algorithm of creative telescoping always terminates with success.

This issue is rooted in the method of creative telescoping by the fact that
sequences are \emph{represented} through the linear recurrence relations that
they satisfy.  Unfortunately, this representation is not very faithful when the
leading terms of the relations vanish for some values of the indices.  The
method of generating functions avoids this issue.  For example, the binomial
coefficient~$\binom{n}{k}$ is represented unambiguously as the coefficient
of~$x^k$ in~$(1+x)^n$ (to be understood as a power series when~$n < 0$), rather
than as a solution to the recurrence relations~$(n-k)\binom{n}{k} =
n\binom{n-1}{k}$ and~$k\binom{n}{k} = n\binom{n-1}{k-1}$.

\subsection*{Related work} The method of generating functions is classical
and has been largely studied, in particular for the approach described here by
\textcite{Egorychev1984}, see also \parencite{EgZi08}. Egorychev's method is a general approach to summation, but not quite an algorithm. In this work, we make it completely effective and practical for the class of binomial sums.

The special case when the generating functions are differentially finite (D-finite) has been studied by \textcite{Lipshitz1989}. From the effectivity point of view, the starting point is his proof that diagonals of D-finite power series are D-finite \parencite{Lip88}. The argument, based on linear algebra, is constructive but does not translate into an efficient algorithm because of the large dimensions involved. This led \textcite[596]{WilfZeilberger1992a} to comment that \emph{``This approach, while it is explicit in principle, in fact yields an infeasible algorithm.''}
Still, using this construction of diagonals, many closure properties of the sequences under consideration (called P-recursive) can be proved (and, in principle, computed). Then, the representation of a convergent definite sum amounts to evaluating a generating series at~1 and this proves the existence of linear recurrences for the definite sums of all P-recursive sequences. \Textcite{AbramovPetkovsek2002} showed that in particular the so-called \emph{proper hypergeometric sequences} are P-recursive in the sense of Lipshitz. The proof is also constructive, relying on Lipshitz's construction of diagonals to perform products of sequences. 

While we are close to Lipshitz's approach, three enhancements make the method
of generating functions presented here efficient: we use more efficient algorithms for
computing multiple integrals and diagonals that have appeared in the last twenty
years \parencite{Chy00,Kou10,BostanLairezSalvy2013,Lai15}; we restrict ourselves to binomial sums, which makes it possible to
manipulate the generating functions through rational integral
representations (see~\S\ref{sec:gfbs} and~\S\ref{sec:residues}); and a third decisive improvement comes with the geometric reduction procedure
for simplifying integral representations (see~\S\ref{sec:geomred}).

Creative telescoping is another summation algorithm developed by \textcite{Zeilberger1991a} and proved to work for all proper hypergeometric sequences \parencite{WilfZeilberger1992a}. This method has the advantage of being applicable and often efficient in practice. However, as already mentioned, it relies on certificates whose size grows fast with the number of variables \parencite{BostanLairezSalvy2013} and, more importantly, whose summation is not straightforward, making the complete automation of the method problematic. For proper hypergeometric sums, a different effective approach developed by \textcite{Takayama95} does not suffer from the certificate problem. It consists in expressing the sum as the evaluation of a hypergeometric series and reducing its shifts with respect to a non-commutative Gr{\"o}bner basis of the contiguity relations of the series, reducing the question to linear algebra in the finite-dimensional quotient. 

The class of sums we consider is a subclass of the sums of proper hypergeometric sequences. We give an algorithm that avoids the computation of certificates in that case, and relies on an efficient method to deal with the integral representations of sums. The same approach has been recently used  by \textcite{BBCHM13} on various examples, though in a less systematic manner.
Examples in \S\ref{sec:applications} give an idea of the extent of the class we are dealing with. It is a subclass of the \emph{balanced multisums}, shown by \textcite{Garoufalidis09} to possess nice asymptotic properties. More recently, a smaller family of \emph{binomial multisums} was studied by \textcite{GaPa14}: they are further constrained to be diagonals of $\mathbb{N}$-rational power series.

\subsection*{Acknowledgments.} We thank Marko Petkov\v{s}ek for orienting us in some of the literature. This work has been supported in part by
FastRelax ANR-14-CE25-0018-01 and by the research grant BU 1371/2-2 of the
Deutsche Forschungsgemeinschaft.  

\section{The algebra of binomial sums}
\label{sec:algbinomsums}

\subsection{Basic objects}
\label{sec:def}

 For all~$n,k\in\Z$, the binomial
 coefficient~$\binom{n}{k}$ is considered in this work as \emph{defined} to be the coefficient of~$x^k$ in the
formal power series~$(1+x)^n$. In other words, 
\[\binom{n}{k}  = %\frac{n^{\underline k}}{k!} = 
\frac{n(n-1)(n-2)\cdots(n-k+1)}{k!} \; \text{for} \; k \geq 0 \quad \text{and}\quad \binom{n}{k}  = 0 \; \text{for} \; k < 0.
\]

For all~$a,b\in \Z$, we define the \emph{directed sum}~$\sum_{k=a}^{\prime \, b}$ as
\[ \ssum_{k=a}^b u_k \eqdef
  \begin{cases}
    \sum_{k=a}^b u_k & \text{if $a\leq b$}, \\
    0 & \text{if $a=b+1$}, \\
    - \sum_{k=b+1}^{a-1} u_k & \text{if $a>b+1$,}
  \end{cases} \]
in contrast with the usual convention that~$\sum_{k=a}^b u_k=0$ when~$a > b$.
This implies the following flexible summation rule for directed sums:
\[ \ssum_{k=a}^b u_k + \ssum_{k=b+1}^c u_k = \ssum_{k=a}^c u_k, \quad \text{for all} \quad a,b,c\in \Z,\] 
and also the geometric summation formula
\[\ssum_{k=a}^{b} r^k = \frac{r^a-r^{b+1}}{1-r} \quad \text{for any} \;  r\neq 1,\]
valid independently of the relative position of~$a$ and~$b$.

Let~$\K$ be a field of characteristic zero, and let $d\geq 1$.  We denote
by~$\cS_d$ the~$\K$-algebra of sequences~$\Z^d\to\K$, where  addition and
 multiplication are performed component-wise.  Elements of~$\Z^d$ are
denoted using underlined lower case letters, such as~$\ul n$.  The algebra~$\cS_d$ may be embedded in the algebra of all
functions~$\Z^{\N}\to\K$ by sending a sequence~$u:\Z^d\to\K$ to the
function~$\tilde u$ defined by
\[ \tilde u(n_1,n_2,\dotsc) = u(n_1,\dotsc,n_d). \]
Let~$\cS$ be the union of the~$\cS_d$'s in the set of all functions~$\Z^{\N}\to\K$.
For~$u\in\cS$ and~$\ul n\in \Z^d$, the notation~$u_{\ul n}$ represents~$u_{\ul n, 0,0,\dotsc}$.
These conventions let us restrict or extend the number of indices as needed without keeping track of it.

\begin{definition}\label{def:bs}
  The algebra of \emph{binomial sums} over the field~$\K$, denoted~$\cB$, is the smallest subalgebra of~$\cS$ such that:
  \begin{enumerate}[(a)]
    \item\label{def:it:delta} The Kronecker delta sequence~$n\in\Z\mapsto\delta_n$, defined by~$\delta_0=1$ and~$\delta_n=0$ if~$n\neq 0$, is in~$\cB$.
    \item\label{def:it:geom} The geometric sequences~$n\in\Z\mapsto C^n$, for all~$C\in\K\setminus\left\{ 0 \right\}$, are in~$\cB$.
    \item\label{def:it:binom} The binomial sequence~$(n,k)\mapsto \binom{n}{k}$ (an element of~$\cS_2$) is in~$\cB$.
    \item\label{def:it:lambda} If~$\lambda : \Z^d\to \Z^e$ is an affine map and if~$u\in\cB$,
      then~$\ul n\in\Z^d\mapsto u_{\lambda(\ul n)}$ is in~$\cB$.
    \item\label{def:it:sum} If~$u\in \cB$, then the following directed indefinite sum is in~$\cB$:
      \[ (\ul n, m)\in \Z^d\times\Z \mapsto \ssum_{k=0}^m u_{\ul n, k}. \]
  \end{enumerate}
\end{definition}

\begin{figure}[tp]
  \centering
\begin{tikzpicture}[level distance=2cm,level 1/.style={sibling distance=3cm},level 2/.style={sibling distance=2cm}]
  \node {$\displaystyle \sum_{k=0}^n \binom{n}{k}\binom{n+k}{k} \ssum_{j=0}^k \binom{k}{j}^3$} [grow'=up]
  child { node {$H_n$} child { node {$\delta_n$} edge from parent node[left] {\ref{def:it:sum}} }   }
  child { node {$\displaystyle \ssum_{k=0}^n \binom{n}{k}\binom{n+k}{k} \ssum_{j=0}^k \binom{k}{j}^3$} 
    child { node {$\displaystyle \ssum_{k=0}^m \binom{n}{k}\binom{n+k}{k} \ssum_{j=0}^k \binom{k}{j}^3$}
      child { node {$\displaystyle \binom{n}{k}\binom{n+k}{k} \ssum_{j=0}^k \binom{k}{j}^3$}
        child { node {$\displaystyle \binom{n}{k}$} }
        child { node {$\displaystyle \binom{n+k}{k}$}
          child { node {$\displaystyle \binom{n}{k}$}
            edge from parent node[right] {\ref{def:it:lambda}} node [left] {$n \mapsto n+k$}
          }
        }
        child {
          node {$\displaystyle \ssum_{j=0}^k \binom{k}{j}^3$}
          child { node {$\displaystyle \ssum_{j=0}^m \binom{k}{j}^3$}
            child { node {$\displaystyle \binom{k}{j}^3$}
              child { node {$\displaystyle \binom{k}{j}$} }
              child { node {$\displaystyle \binom{k}{j}$} }
              child { node {$\displaystyle \binom{k}{j}$} }
              edge from parent node[left] {\ref{def:it:sum}}
            }
            edge from parent node[left] {\ref{def:it:lambda}} node [right] {$m \mapsto k$}
          }
        }
        edge from parent node[left] {\ref{def:it:sum}}
      }
      edge from parent node[left] {\ref{def:it:lambda}}  node [right] {$m \mapsto n$}
    } };
\end{tikzpicture}
  
  \caption{Construction of a binomial sum. The last step replaces the directed sum
    by a usual sum in order to have nonzero values for~$n\geq 0$ only.}
  \label{fig:binomtree}
\end{figure}

\label{sec:bsexamples}
Let us give a few useful examples.  All polynomial sequences~$\Z^d\to\K$ are
in~$\cB$, because~$\cB$ is an algebra and because the sequence~$n\mapsto \binom{n}{1}=n$ is in~$\cB$, thanks to
points~\ref{def:it:binom} and~\ref{def:it:lambda} of the definition.

Let~$(H_n)_{n\in \Z}$ be the sequence defined by~$H_n = 1$ if~$n\geq 0$ and~$H_n=0$ if $n<0$.
It is a binomial sum since~$H_n = \sum_{k=0}^{\prime\,n}\delta_k$.
As a consequence, we obtain the closure of~$\cB$ by usual (indefinite) sums since
\[ \sum_{k=0}^m u_{\ul n, k} = H_m\ssum_{k=0}^m u_{\ul n,k}. \]

By combining the rules~\ref{def:it:lambda} and~\ref{def:it:sum} of the
definition, we also obtain the closure of~$\cB$ under sums whose bounds depend
linearly on the parameter: if~$u\in\cB$ and if~$\lambda$ and~$\mu$ are affine
maps~$:\Z^d\to\Z$, then the sequence
\begin{equation}
 \ul n\in\Z^d \mapsto \ssum_{k=\lambda(\ul n)}^{\mu(\ul n)} u_{\ul n,k}
 \label{eqn:generalizedsummation}
\end{equation}
is a binomial sum.

See also Figure~\ref{fig:binomtree} for an example of a classical binomial sum.

\subsection{Characterization of binomial sums}
\label{sec:caracbs}

A few simple criteria make it possible to prove that a given sequence is not a binomial sum.
For example:
\begin{itemize}
  \item The sequence~$(n!)_{n\geq 0}$ is not a binomial sum. Indeed, the set of
    sequences that grow at most exponentially is closed under the rules
    that define binomial sums. Since~$\binom{n}{k} \leq 2^{|n| + |k|}$, 
    every binomial sum grows at most exponentially; but this is not the case for the
    sequence~$(n!)_{n\geq 0}$.
  \item The sequence of all prime integers is not a binomial sum
    because it does not satisfy any nonzero linear recurrence relation with polynomial
    coefficients \parencite{FlaSalGer05}, whereas every binomial sum does, see
    Corollary~\ref{cor:bs-dfinite}.  
  \item The sequence~$(1/n)_{n\geq 1}$ is not a binomial sum. To prove this, we
    can easily reduce to the case where~$\K$ is a number field and then study
    the denominators that may appear in the elements of a
    binomial sum. One may introduce new prime divisors in the denominators only by
    multiplying with a scalar or with rule~\ref{def:it:geom}, so that the
    denominators of the elements of a given binomial sum contain only finitely
    many prime divisors. This is clearly not the case for the sequence~$(1/n)_{n\geq 1}$.

    By the same argument, the first sum of Eq.~\eqref{eqn:sum-exem-intro} is not a binomial sum. Indeed, by creative telescoping, it can be shown to equal ${(2n+1)4^{n+1}}/{\binom{2n+2}{n+1}}+{1}/{3}$
    and thus all prime numbers appear as denominators.
  \item The sequence $(u_n)_{n\geq 0}$ defined by $u_0=0, u_1=1$ and by the recurrence 
$(2n+1)u_{n+2} - (7n+11)u_{n+1} +(2n+1)u_n =0$ is not a binomial sum.
This follows from the asymptotic estimate $u_n \sim C \cdot (4/(7-\sqrt{33}))^n \cdot n^{\sqrt{75/44}}$, with~$C \approx 0.56$, and the fact that $\sqrt{75/44}$ is not a rational number \parencite[Theorem 5]{Garoufalidis09}.

\end{itemize}

These criteria are, in substance, the only ones that we know to prove that a
given sequence is not a binomial sequence.  Conjecturally, they characterize univariate
binomial sums. Indeed, we will see that the equivalence between univariate binomial sums and
diagonals of rational functions (Theorem~\ref{thm:main}) leads, among many interesting corollaries, to the following reformulation of 
a conjecture due to \textcite[Conjecture~4]{Christol1990}: “any sequence $(u_n)_{n\geq 0}$ of integers
that grows at most exponentially and that is solution of a linear recurrence
equation with polynomial coefficients is a binomial sum.” 
% The equivalence between diagonals and binomial sums has many interesting corollaries; it will be addressed in~\S\ref{sec:diagonals}.

%\subsection{Binomial sums and diagonals of rational functions}
%The generating functions of binomial sums are precisely the formal power series
%obtained by extracting the \emph{diagonal} of multivariate power series.  The
%notation~$z_{1\smile d}$ stands for~$z_1,\dotsc,z_d$.
%
%\begin{theorem}
%  A sequence~$u : \N\to \K$ is a binomial sum if and only if 
%  there exists a rational power series~$\sum_{\ul n \in \N^d} a_{\ul n} \ul z^{\ul n} \in \K\llbracket z_{1\smile d} \rrbracket \cap \K(z_{1\smile d})$
%  such that~$u_k = a_{k,\dotsc,k}$ for all~$k\geq 0$.
%  \label{theo:diag-equiv}
%\end{theorem}
%
%The proof is given in~\S\ref{sec:diagonals}.
%The theory of diagonals of rational functions provides numerous corollaries.
%They will be detailed also in~\S\ref{sec:diagonals}.
%
%\todo{Show as a corollary that $(n!)_{n \geq 0}$ is not a binomial sum}

\section{Generating functions}
\label{sec:gfun}

To go back and forth between sequences and power series, one can use convergent
power series and Cauchy's integrals to extract coefficients, as shown in the
introduction with Dixon's identity, or one can use formal power series.  Doing
so, the theory keeps close to the actual algorithms and we avoid the tedious
tracking of convergence radii.  However, this requires the introduction of a
\emph{field} of multivariate formal power series that makes it possible, by
embedding the rational functions into it, to define what the coefficient of a
given monomial is in the power series expansion of an arbitrary multivariate
rational function.  We choose here to use the field of \emph{iterated Laurent
series}. It is an instance of the classical field of Hahn series when the value
group is~$\Z^n$ with the lexicographic order. We refer to \textcite{Xin04} for
a complete treatment of this field and we simply gather here the main
definitions and results.

\subsection{Iterated Laurent series}
\label{sec:laurent}

For a field~$\A$, let~$\A\slf{t}$ be the field of univariate formal Laurent power series
with coefficients in~$\A$.  %For~$n\in \Z$ and~$f\in \A\slf{t}$, let~$[t^n]f$
%denote the coefficient of~$t^n$ in~$f$.
For~$d \geq 0$, let~$\L_d$ be the
field of iterated formal Laurent power
series~$\K\slf{z_d}\slf{z_{d-1}}\dotsm\slf{z_1}$.
It naturally contains the field of rational functions in~$z_1,\dotsc,z_d$.
For~$\ul n = (n_{1},\dotsc,n_{d})
\in \Z^d$, let~$\ul z^{\ul n}$ denote the monomial~$z_1^{n_1}\dotsm
z_d^{n_d}$, and for~$f\in\L_d$, let~$[\ul z^{\ul n}] f$ denote the coefficient
of~$\ul z^{\ul n}$ in~$f$, that is the coefficient of~$z_d^{n_d}$ in the coefficient of~$z_{d-1}^{n_{d-1}}$ in [\dots] in the coefficient of~$z_1^{n_1}$ of~$f$.
An element of~$f\in\L_d$ is entirely characterized by its coefficient
function~$\ul n\in\Z^d \mapsto [\ul z^{\ul n}] f$.
On occasion, we will write~$[z_1^n]f$, this means~$[z_1^n z_2^0 \dotsm z_d^0]f$, e.g., $[1] f$ means~$[z_1^0 \dotsm z_d^0]f$; more generally we write~$[m] f$, where~$m$ is a monomial.
In any case, the bracket notation always yields an element of~$\K$.

Let~$\prec$ denote the lexicographic ordering on~$\Z^d$.  For~$\ul n,\ul
m\in\Z^d$ we write~$\ul z^{\ul n} \prec \ul z^{\ul m}$ if~$\ul m \prec \ul n$
(mind the inversion: a monomial is larger when its exponent is smaller).  In
particular~$z_1\prec \dotsb \prec z_d \prec 1$.  With an analytic point of
view, we work in an infinitesimal region where~$z_d$ is infinitesimally small
and where~$z_i$ is infinitesimally small compared to~$z_{i+1}$. The
relation~$\ul z^{\ul n} \prec \ul z^{\ul m}$ means that~$|\ul z^{\ul n}|$
is smaller than~$|\ul z^{\ul m}|$ in this region.

For~$f\in\L_d$, the \emph{support} of~$f$, denoted~$\mop{supp}(f)$, is the
set
of all~$\ul n\in\Z_d$ such that~$[\ul z^{\ul n}]f$ is not zero. 
It is well-known \parencite[e.g.][Prop.~2-1.2]{Xin04} that a function~$\varphi:\Z^d\to\K$ is
the coefficient function of an element of~$\L_d$ if and only
if the support of~$\varphi$ is well-ordered for the order~$\prec$ (that is to say every subset
of the support of~${\varphi}$ has a least element).
The \emph{valuation} of a non zero~$f \in \L_d$, denoted~$v(f)$, is the smallest element
of~$\mop{supp}(f) \subset \Z^d$ for the ordering~$\prec$.  Since~$\mop{supp}(f)$ is
well-ordered, $v(f)$ does exist.  The \emph{leading monomial} of~$f$,
denoted~$\lm(f)$, is~$\ul z^{v(f)}$; it is the largest monomial that appears
in~$f$.

For~$1\leq i \leq d$, the partial derivative~$\partial_i =
\partial/\partial z_i$ with respect to the variable~$z_i$, defined for
rational functions, extends to a derivation in~$\L_d$ such that~$[\ul z^{\ul
n}] \partial_i f = (n_i+1) [\ul z^{\ul n}](f/z_i)$ for any~$f\in \L_d$.

Let~$w_1,\dotsc,w_d$ be monomials in the variables~$z_1,\dotsc,z_e$. When~$f$ is a
rational function, the substitution~$f(w_1,\dotsc,w_d)$ is defined in a
simple way.  For elements of~$\L_d$, it is slightly more technical.  Let~$\varphi : \Z^d
\to \Z^e$ be the additive map defined by~$\varphi(\ul n) = v(\ul w^{\ul n})$ (recall that $v$ denotes the valuation).
Conversely, this map entirely determines the monomials~$w_1,\dotsc,w_d$.
When~$\varphi$ is strictly increasing (and thus injective) with
respect to the lexicographic ordering,
we define~$f^\varphi$ as the unique element of~$\L_e$ such
that
\[ [\ul z^{\ul n}] f^{\varphi} = 
\begin{cases}
  [\ul z^{\varphi^{-1}(\ul n)}] f & \text{if $\ul n\in \varphi(\Z^d)$}, \\
  0 & \text{otherwise.}
\end{cases}
\]
The map~$f\in\L_d \mapsto f^\varphi\in\L_e$ is a field morphism.
In particular, for a monomial~$\ul z^{\ul n}$, we check that~$\left( \ul z^{\ul n} \right)^\varphi = \ul z^{\varphi(\ul n)}$.
When~$f$ is a rational function, $f^\varphi$ coincides with the usual
substitution~$f(w_{1}, \ldots, w_d)$.

Another important construction is the sum of geometric sequences.  Let~$f\in\L_d$ be a
Laurent power series with~$\lm(f) \prec 1$.  The set of all~$\ul n\in \Z^d$
such that there is at least one~$k\in\N$ such that~$[\ul z^{\ul n}] f^k \neq 0$ is well-ordered \parencite[Theorem~3.4]{Neu49}.  Moreover, for any~$\ul n\in \Z^d$, the
coefficient~$[\ul z^{\ul n}] f^k$ vanishes for all but finitely
many~$k\in\N$ \parencite[Theorem~3.5]{Neu49}.
The following result is easily deduced.
\begin{lemma}
  \label{lem:geomsum}
    Let~$f\in\L_d$ and let~$\ul z^{\ul n}$ be a monomial. If $\lm(f)\prec 1$,
  then~$[\ul z^{\ul n}] f^k = 0$ for all but finitely many~$k\in\N$ and
  moreover
  \[ [\ul z^{\ul n}] \frac{1}{1-f} = \sum_{k\geq 0} [\ul z^{\ul n}] f^k. \]
\end{lemma}

In what follows, there will be variables~$z_{1},\dotsc,z_d$, denoted~$z_{1\smile d}$ (and sometimes under
different names) and we will denote~$z_1\prec\dotsb\prec z_d$ the fact that we
consider the field~$\L_d$ with this ordering to define coefficients, residues,
etc.  The variables are always ordered by increasing index, and~$t_{1\smile
d} \prec z_{1\smile e}$ denotes~$t_1\prec\dotsb\prec t_d\prec z_1
\prec\dotsb\prec z_e$.  An element of~$\L_d\cap \K(z_{1\smile d})$ is called a
\emph{rational Laurent series}, and an element of~$\K\llbracket z_{1\smile d}
\rrbracket \cap \K(z_{1\smile d})$ is called a \emph{rational power series}.

\begin{example}
Since~$\L_d$ is a field containing all the rational functions in the variables~$z_{1},\dotsc,z_d$,
one may define the \emph{coefficient of a monomial in a rational function}.
However, it strongly depends on the ordering of the variables.
For example, in~$\L_2 = \K\slf{z_2}\slf{z_1}$, the coefficient of~$1$ in~$z_2/(z_1+z_2)$ is~$1$ because
\[ \frac{z_2}{z_1+z_2} = 1 - \frac{1}{z_2} z_1 + \mathcal O(z_1^2), \]
whereas the coefficient of~$1$ in~$z_1/(z_1+z_2)$ is~$0$ because
\[ \frac{z_1}{z_1+z_2} = \frac{1}{z_2}z_1 + \mathcal O(z_1^2). \]
\end{example}

\subsection{Binomial sums and coefficients of rational functions}
\label{sec:gfbs}

Binomial sums can be obtained as certain extractions of coefficients of
rational functions.  We will make use of sequences of the form $u : \ul
n\in\Z^d \mapsto [1]\left( R_0 R_1^{n_1}\dotsm R_d^{n_d} \right)$,
where~$R_{0\smile d}$ are rational functions of ordered variables~$z_{1\smile
r}$.  They generalize several notions.  For example if~$R_0 \in
\K(z_1,\dotsc,z_r)$ and~$R_i = 1/z_i$, then~$u$ is simply the coefficient
function of~$R_0$; and if~$d=1$, $R_0 \in \K(z_1,\dotsc,z_r)$ and~$R_1 =
1/(z_1\dotsm z_r)$, then~$u$ is the sequence of the diagonal coefficients
of~$R_0$.

\begin{theorem}\label{prop:bs-ct}
  Every binomial sum is a linear combination of finitely many sequences of the form
  $\ul n\in\Z^d \mapsto [1]\left( R_0 R_1^{n_1}\dotsm R_d^{n_d} \right)$,
  where~$R_{0\smile d}$ are rational functions of ordered variables~$z_{1\smile r}$, for some~$r\in \N$.
\end{theorem}

\begin{proof}
It is clear that~$\delta_n = [1] z^n$, that~$C^n = [1] C^n$ and
that~$\binom{n}{k}=[1](1+z)^n z^{-k}$, for all~$n,k\in\Z$.  Thus, it is enough
to prove that the vector space generated by the sequences of the
form~$[1]\left( R_0 R_1^{n_1}\dotsm R_d^{n_d} \right)$ is a subalgebra of~$\cS$
which is closed by the rules defining binomial sums, see~\S\ref{sec:def}.

First, it is closed under product: if~$R_{0\smile d}$ and~$R'_{0\smile d}$ are rational functions in the variables~$z_{1\smile r}$ and~$z'_{1\smile r'}$ respectively, then 
\[ [1]\left( R_0 \prod_{i=1}^d R_i^{n_i} \right) [1]\left( R'_0 \prod_{i=1}^d {R'_i}^{n_i} \right) = [1]\left( R_0 R'_0 \prod_{i=1}^d (R_i R'_i)^{n_i} \right), \]
with the ordering~$z_{1\smile r} \prec z'_{1\smile\smash{r'}}$,
for example.
Then, to prove the closure under change of variables, it is enough to reorder the factors:
\[ [1]\left( R_0 \prod_{i=1}^d R_i^{\sum_{j=1}^e a_{ij}n_j + b_i} \right)
= [1]\left( \left( R_0 \prod_{i=1}^d R_i^{b_i} \right) \prod_{j=1}^e \left( \prod_{i=1}^d R_i^{a_{ij}} \right)^{n_j}  \right). \]

Only the closure under partial sum remains.
Let~$u$ be a sequence of the form
\[ (\ul n, k) \in \Z^{d}\times\Z \mapsto [1]\left( T^k R_0 \prod_{i=1}^{d} R_i^{n_i} \right), \]
where~$T$ and~$R_{0\smile d}$ are rational functions.
If~$T=1$, then
\[ \ssum_{k=0}^m u_{\ul n,k} = [1]\left( (m+1) R_0 \prod_{i=1}^{d}  R_i^{n_i} \right) = [1]\left( \frac{(1+v)R_0}{v} (1+v)^{m} \prod_{i=1}^{d}  R_i^{n_i} \right), \]
where~$v$ is a new variable, because~$[v](1+v)^{m+1} = m+1$.
If~$T\neq 1$, then
\[ \ssum_{k=0}^m u_{\ul n,k} = [1]\left( \frac{R_0}{1-T} \prod_{i=1}^{d} R_i^{n_i} \right)
- [1]\left( \frac{R_0 T}{1-T} T^{m} \prod_{i=1}^{d} R_i^{n_i} \right), \]
which concludes the proof.
\end{proof}

%\begin{corollary}\label{coro:bs-ct-Nd}
%  Every binomial sum with support in~$\N^d$ is a linear combination of finitely
%  many sequences with support in~$\N^d$ of the form $\ul n\in\Z^d \mapsto
%  [1]\left( R_0 R_1^{n_1}\dotsm R_d^{n_d} \right)$, where~$R_{0\smile d}$ are
%  rational functions of ordered variables~$z_{1\smile r}$.
%\end{corollary}
%
%\begin{proof}
%  Let~$u :\Z^d \to \K$ be a binomial sum with~$\mop{supp} u\subset\N^d$.
%  By Theorem~\ref{prop:bs-ct},
%  $u$ is a linear combination of sequences of the form
%  $\ul n\in\Z^d \mapsto [1]\left( R_0 R_1^{n_1}\dotsm R_d^{n_d} \right)$,
%   where~$R_{0\smile d}$ are
%  rational functions of ordered variables~$z_{1\smile r}$.
%  The support of each term of the combination, considered separately,
%  may not be contained in~$\N^d$.
%  This can be fixed thanks to the formula
%  \begin{multline*}
%    [1]\left( \frac{R_0}{\prod_{i=1}^{d} (1-t_{i})} (t_{1} R_1)^{n_1} \dotsm (t_{d}R_d)^{n_d} \right) =  \\
%      \begin{cases}
%        [1]\left( R_0 R_1^{n_1}\dotsm R_d^{n_d} \right) & \text{if $n_1, \dotsc, n_d \geq 0$}\\
%        0 &\text{otherwise,}
%      \end{cases}
%  \end{multline*}
%  where~$t_{1\smile d}$ are new variables with~$t_{1\smile d} \prec z_{1\smile r}$.
%\end{proof}
%
The previous proof is algorithmic and Algorithm~\ref{algo:sumtoct} p.~\pageref{algo:sumtoct}
proposes a variant which is more suitable for actual computations---see~\S\ref{sec:gfalgos}.
It is a straightforward rewriting
procedure and often uses more terms and more variables than necessary.  The
geometric reduction procedure, see~\S\ref{sec:geomred}, handles this issue.

\begin{example}\label{ex:sumtoct}
  Following step by step the proof of Theorem~\ref{prop:bs-ct}, we obtain that
  \begin{multline*}
    \sum_{k=0}^n \binom{n}{k}\binom{n+k}{k} = 
    [1] \frac{1+z_2}{1+z_1-z_1z_2} \left( \frac{(1+z_1)(1+z_2)^2}{z_1 z_2} \right)^n \\- [1] \frac{z_1z_2}{1+z_1-z_1z_2} \left( (1+z_1)(1-z_2)  \right)^n. 
  \end{multline*}
  We notice that the second term is always zero.
\end{example}

\begin{corollary}\label{cor:bs-coeff}
  For any binomial sum~$u : \Z^d\to \K$, there exist $r\in\N$ and a rational
  function~$R(t_{1\smile d}, z_{1\smile r})$, with~$t_{1\smile d}\prec
  z_{1\smile r}$, such that for any~$\ul n\in\Z^d$,
  \[ [\ul t^{\ul n} \ul z^{\ul 0}] R =
  \begin{cases}
    u_{\ul n} & \text{if } n_1,\dotsc,n_d \geq 0,\\
    0 & \text{otherwise.}
  \end{cases} \]
\end{corollary}
(See Proposition~\ref{converse-coro} below for a converse.)
\begin{proof}
  It is enough to prove the claim for a generating set of the vector space of
  all binomial sums.  In accordance with
  Theorem~\ref{prop:bs-ct}, let~$u : \Z^d \to \K$ be a sequence 
  of the form
  $\ul n\in\Z^d \mapsto [1]\left( S_0 S_1^{n_1}\dotsm S_d^{n_d} \right)$,
  where $S_{0\smile d}$ are rational functions of ordered variables~$z_{1\smile r}$.
  Let~$R(t_{1\smile d}, z_{1\smile r})$ be the rational function
  \[ R = S_0 \prod_{i=1}^d \frac{1}{1-t_i S_i}. \]
  Lemma~\ref{lem:geomsum} implies that for any~$\ul n\in \Z^d$
  \[ [\ul t^{\ul n} \ul z^{\ul 0}] R = \sum_{k_1,\dotsc,k_d \geq 0} [\ul t^{\ul n}  \ul z^{\ul 0}]\left( S_0 \cdot (t_1 S_1)^{k_1}\dotsm (t_d S_d)^{k_d}  \right). \]
  Since the variables~$t_i$ do not appear in the~$S_i$'s, the coefficient in the sum is not zero only if~$\ul n = \ul k$.
  In particular, if~$\ul n$ is not in~$\N^d$, then~$[\ul t^{\ul n} \ul z^{\ul 0}] R = 0$.
  And if~$\ul n\in\N^d$, then
  \[ [\ul t^{\ul n} \ul z^{\ul 0}] R =  [\ul t^{\ul n}  \ul z^{\ul 0}]\left( S_0 \cdot (t_1 S_1)^{n_1}\dotsm (t_d S_d)^{n_d}  \right) =  [1]\left( S_0 S_1^{n_1}\dotsm S_d^{n_d} \right) = u_{\ul n}, \]
  which concludes the proof.
\end{proof}

%\begin{example}
%  Continuing Example~\ref{ex:sumtoct} and following the steps of the proof of Corollary~\ref{cor:bs-coeff},
%  we compute that
%  \[ \sum_{k=0}^n \binom{n}{k}\binom{n+k}{k} = [t^n z_1^0 z_2^0]\left(
%  \frac{z_1z_2}{(z_1z_2 - (1+z_2)^2(1+z_1)t)(1-(1+z_2)(1+z_1)t)}
%  \right). \]
%\end{example}

\subsection{Residues}\label{sec:residues}
The notion of residue makes it possible to represent the full generating function of a binomial sum.
It is a key step toward their computation.
For~$f \in \L_d$ and~$1\leq i \leq d$, the \emph{formal residue} of~$f$ with
respect to~$z_i$, denoted~$\res_{z_i} f$, is the unique element of~$\L_d$ such
that
\[ [\ul z^{\ul n}](\res_{z_i} f) = 
\begin{cases}
  [\ul z^{\ul n}]( z_i f) & \text{if $n_i=0$}, \\
  0 & \text{otherwise.}
\end{cases}
\]
This is somehow the coefficient of~$1/z_i$ in~$f$, considered as a power series
in~$z_i$ (though~$f$ is not a Laurent series in~$z_i$ since it may contain
infinitely many negative powers of~$z_i$).  When~$f$ is a rational function, care should
be taken not to confuse the formal residue with the \emph{classical residue}
at~$z_i=0$, which is the coefficient of~$1/z_i$ in the partial fraction
decomposition of~$f$ with respect to~$z_i$.  However, the former can be expressed in terms of
classical residues, see~\S\ref{sec:geomred}, and like the classical residues, it
vanishes on derivatives:
\begin{lemma}
  $\res_{z_i} \partial_i f = 0$ for all~$f\in\L_d$.
  \label{lem:resder}
\end{lemma}
If~$\alpha$ is a set of variables~$\left\{ z_{i_1},\dotsc,z_{i_r} \right\}$,
let~$\res_\alpha f$ denote the iterated residue $\res_{z_{i_1}}\dotsm\res_{z_{i_r}} f$.
It is easily checked that this definition does not depend on the order in which the variables appear.
This implies, together with Lemma~\ref{lem:resder}, the following lemma:
\begin{lemma}
  $\res_\alpha\left( \sum_{v\in \alpha} \partial_v f_v \right) = 0$, for any family~$(f_v)_{v\in\alpha}$ of elements of~$\L_d$.
  \label{lem:resders}
\end{lemma}

The following lemma will also be useful:
\begin{lemma}
  Let~$\alpha,\beta \subset \left\{ z_{1\smile d} \right\}$ be disjoint sets of variables.
  If~$f\in \L_d$ does not depend on the variables in~$\beta$ and if~$g\in\L_d$ does not depend on the variables in~$\alpha$, then
  \[ \big( \res_\alpha f \big)\big( \res_\beta g \big) = \res_{\alpha\cup\beta}(fg). \]
  \label{lem:ressepvar}
\end{lemma}
\begin{proof}
  $\res_{\alpha\cup\beta} fg = \res_\alpha \res_\beta fg = \res_\alpha\big( f \res_\beta g\big) = \big( \res_\alpha f \big)\big( \res_\beta g \big)$.
\end{proof}

The order on the variables matters.  For example, let~$F(x,y)$ be a
rational function, with~$x\prec y$.  The residue~$\res_x F$ is a rational
function of~$y$. Indeed, if~$-n$ is the exponent of~$x$ in~$\lm(F)$, which we
assume to be negative, then
\[ \res_x F = \frac{1}{(n-1)!} \left. \left(\frac{\partial^{n-1}}{\partial x^{n-1}} x^n F \right) \right|_{x= 0}. \]
In contrast, the residue with respect to~$y$ (or any other variable which is
not the smallest) is not, in general, a rational function.  It is an important point
 because we will represent generating series of binomial
sums---which need be neither rational nor algebraic---as residues of
rational functions.

\begin{example}
  Let~$F = \frac{1}{xy(y^2+y-x)}$.
  Assume that~$y \prec x$, that is to say, we work in the field~$\K\slf{x}\slf{y}$.
  We compute that~$\res_y F = -\frac{1}{x^2}$
  because
  \[ F = -\frac{1}{x^2 y} + \frac{(y+1)}{x^2 (y^2+y-x)} \]
  and the second term does not contain any negative power of~$y$.

  On the other hand, if we assume that~$x \prec y$, then~$\res_y F$ is not a rational function anymore.
  Using Proposition~\ref{prop:geomred}, we can compute that
  \[ \res_y F = -\frac{1}{x^2} + \frac{2}{x + 4x^2 - x(1+4 x)^{1/2}} = -\frac1x + 3 - 10x + 35 x^2 + \mathcal O(x^3). \]
\end{example}

Residues of rational functions can be used to represent any binomial sum; it is
the main point of the method and it is the last corollary of
Theorem~\ref{prop:bs-ct}.  Following Egorychev we call them \emph{integral
representations}, or \emph{formal} integral representations, to emphasize the
use of formal power series and residues rather than of analytic objects.

\begin{corollary}
  [Integral representations]
  \label{coro:bs-res}
  For any binomial sum~$u : \Z^d\to \K$, there exist~$r\in\N$ and a rational
  function~$R(t_{1\smile d}, z_{1\smile r})$, with~$t_{1\smile d}\prec
  z_{1\smile r}$, such that
  \[ \sum_{\ul n \in \N^d} u_{\ul n} \ul t^{\ul n} = \res_{z_{1\smile r}} R. \]
\end{corollary}
  In other words, the generating function of the restriction to~$\N^d$ of a binomial sum is a residue of a rational function.

\begin{proof}
  By Corollary~\ref{cor:bs-coeff}, there exist $r\in\N$ and a rational function 
  ${\tilde R}(t_{1\smile d}, z_{1\smile r})$ such that~$u_{\ul n}=[\ul t^{\ul n} \ul z^{\ul 0}] {\tilde R}$ for all~$\ul n\in \N^d$. 
Let $R$ be ${\tilde R}/(z_1\dotsm z_r)$. Then
  \[ \sum_{\ul n \in \N^d} u_{\ul n} t^{\ul n} = 
    \sum_{\ul n \in \N^d} \left( [\ul t^{\ul n} \ul z^{\ul 0}] {\tilde R} \right) \ul t^{\ul n} =
    \res_{z_{1\smile r}}R. \qedhere \]
\end{proof}

In \S\ref{sec:diagonals}, we prove an equivalence between univariate binomial sums and
\emph{diagonals of rational functions} that are a special kind of residues.

\subsection{Algorithms}
\label{sec:gfalgos}

\begin{algo}[t]
  \begin{algorithmic}
    \Function{SumToCT}{$u$}
    \If{$u_{\ul n} = \delta_{n_1}$}  \Return $z_1^{n_1}$
    \ElsIf{$u_{\ul n} = a^{n_1}$ for some~$a\in\K$}
      \Return $a^{n_1}$
    \ElsIf{$u_{\ul n} = \binom{n_1}{n_2}$}
      \Return $(1+z_1)^{n_1} /z_1^{n_2}$
    \ElsIf{$u_{\ul n} = v_{\ul n} + w_{\ul n}$}
      \Return $\textsc{SumToCT}(v)+ \textsc{SumToCT}(w)$
      % \State $P(n_{1\smile d} ; z_{1\smile r}) \gets \textsc{SumToCT}(v)$
      % \State $Q(n_{1\smile e} ; z_{1\smile s}) \gets \textsc{SumToCT}(w)$
      % \State \Return $P + Q$
    \ElsIf{$u_{\ul n} = v_{\ul n} w_{\ul n}$}
      \State $P(n_{1\smile d} ; z_{1\smile r}) \gets \textsc{SumToCT}(v)$
      \State $Q(n_{1\smile e} ; z_{1\smile s}) \gets \textsc{SumToCT}(w)$
      \State \Return $P(n_{1\smile d} ; z_{1\smile r}) Q(n_{1\smile e} ; z_{r+1\smile r+s})$
    \ElsIf{$u_{\ul n} = v_{\lambda(\ul n)}$ for some affine map~$\lambda : \Z^d\to\Z^e$}
      \State $P(n_{1\smile d} ; z_{1\smile r}) \gets \textsc{SumToCT}(v)$
      \State \Return $P(\lambda(n_{1\smile d}) ; z_{1\smile r})$
    \ElsIf{$u_{n_1,n_2,\dotsc} = \ssum_{k=0}^{\ n_1} v_{k,n_2,\dotsc}$}
      \State $P(n_{1\smile d} ; z_{1\smile r}) \gets \textsc{SumToCT}(v)$    
      \State Compute~$Q(n_{1\smile d} ; z_{1\smile r})$ s.t.~$Q(n_{1}+1, n_{2\smile d} ; z_{1\smile r})-Q(n_{1\smile d} ; z_{1\smile r}) = P$
      \State
      \Comment See Lemma~\ref{lem:polygeomsum}.
      \State \Return $Q(n_{1}+1, n_{2\smile d} ; z_{1\smile r}) - Q(0, n_{2\smile d} ; z_{1\smile r})$
    \EndIf
    \EndFunction
  \end{algorithmic}
  \caption[Constant term representation of a binomial sum]{Constant term representation of a binomial sum
\begin{description}
  \item[Input] A binomial sum~$u:\Z^d\to \K$ given as an abstract syntax tree 
  \item[Output] $P(n_{1\smile d}; z_{1\smile r})$ a linear combination of expressions of the form~$\ul n^{\ul \alpha} R_0 R_1^{n_1}\dotsm R_d^{n_d}$ where~$\ul\alpha\in\N^d$ and~$R_{0\smile d}$ are rational functions of~$z_{1\smile r}$.
  \item[Specification]  $u_{\ul n} = [1] P(n_{1\smile d} ; z_{1\smile r})$ for any~$\ul n\in \Z^d$.
\end{description}
  }
\label{algo:sumtoct}
\end{algo}

\begin{algo}[t]
  \begin{algorithmic}
    \Function{SumToRes}{$u$}
    \State $\displaystyle \sum_{k=1}^m \ul n^{\ul \alpha_k} R_{k,0} R_{k,1}^{n_1} \dotsm R_{k,d}^{n_d} \gets \textsc{SumToCT}(u)$
    \State \Return $\displaystyle \frac{1}{z_1\dotsc z_r}\sum_{k=1}^m R_{k,0} F_{\alpha_1}(t_1 R_{k,1}) \dotsm F_{\alpha_d}(t_d R_{k,d})$
    \State \hfill where $F_0(z) = \frac{1}{1-z}$ and~$F_{\alpha+1}(z) =  z F'_\alpha(z)$
    \EndFunction
  \end{algorithmic}
  \caption[Integral representation of a binomial sum]{Integral representation of a binomial sum
\begin{description}
  \item[Input] A binomial sum~$u:\Z^d\to \K$ given as an abstract syntax tree 
  \item[Output] A rational function $R(t_{1\smile d}, z_{1\smile r})$
  \item[Specification] $\sum_{\ul n \in \N^d} u_{\ul n} t^{\ul n} = \res_{z_{1\smile r}} R$, where~$t_{1\smile d}\prec z_{1\smile r}$
\end{description}
  }
\label{algo:sumtores}
\end{algo}

The proofs of Theorem~\ref{prop:bs-ct} and its corollaries are constructive, but in order to implement them, it is useful to consider
not only sequences of the form~$[1](R_0 \ul R^{\ul n})$ but also sequences of the form~$[1]( \ul n^{\ul \alpha} R_0 \ul R^{\ul n})$,
for some rational functions~$R_{0\smile d}$ and~$\ul \alpha\in \N^d$.
They bring two benefits. First, polynomial factors often appear in binomial sums and it is convenient to be able to represent them without adding new variables. Second, we can always perform sums
without adding new variables, contrary to what is done in the case~$T=1$ of the proof of Theorem~\ref{prop:bs-ct}.

\begin{lemma}
  \label{lem:polygeomsum}
  For any~$A \neq 0$ in some field~$k$, and any~$\alpha \in \N$, there exists a polynomial $P\in k[n]$ of degree at most~$\alpha+1$ such that for all~$n\in\Z$,
  \[ n^\alpha A^n = P(n+1) A^{n+1} - P(n) A^n. \]
\end{lemma}
\begin{proof}
  It is a linear system of~$\alpha+2$
  equations, one for each monomial~$1, n, \dotsc, n^{\alpha+1}$, in the~$\alpha+2$ coefficients of~$P$.
  If~$A\neq 1$,
  the homogeneous equation~$P(n+1) A = P(n)$ has no solutions, so the system is
  invertible and has a solution.  If~$A = 1$, then the equation~$P(n+1) A =
  P(n)$ has a one-dimensional space of solutions but the equation corresponding
  to the monomial~$n^{\alpha+1}$ reduces to~$0=0$, so the system again has a
  solution.
\end{proof}

Algorithm~\ref{algo:sumtoct} implements Theorem~\ref{prop:bs-ct}
in this generalized setting.
The input data, a binomial sum, is given as an expression considered as an abstract syntax tree, as depicted in Figure~\ref{fig:binomtree}.
The implementation of Corollaries~\ref{cor:bs-coeff} and~\ref{coro:bs-res}
must also be adjusted to handle the monomials in~$\ul n$.

To this effect, for~$\alpha\in\N$, let
\[ F_\alpha(z) \eqdef  \sum_{n\geq 0}n^\alpha z^n = \left( z\frac{\ud}{\ud z} \right)^\alpha \cdot \frac{1}{1-z}. \]
Now let us consider a sequence $u : \ul n\in\Z^d \mapsto [1]\left( \ul n^{\ul \alpha} S_0 S_1^{n_1}\dotsm S_d^{n_d} \right)$,
where $S_{0\smile d}$ are rational functions of ordered variables~$z_{1\smile r}$ and~$\ul \alpha\in \N^d$
Let
\[ \tilde R = S_0 \prod_{i=1}^d\frac{1}{1-t_i S_i} \quad\text{and}\quad  R = S_0 \prod_{i=1}^d F_{\alpha_i}(t_i S_i), \]
where~$t_{1\smile d} \prec z_{1\smile r}$.
By definition of the~$F_\alpha$'s, we check that
\[ R = \left( t_1\frac{\partial}{\partial t_1} \right)^{\alpha_1} \dotsm \left( t_d\frac{\partial}{\partial t_d} \right)^{\alpha_d} \cdot \tilde R,\]
so that~$[\ul t^{\ul n} \ul z^{\ul 0} ] R = \ul n^{\ul \alpha} [\ul t^{\ul n} \ul z^{\ul 0} ] \tilde R$, for~$\ul n\in\Z^d$,
by definition of the differentiation in~$\L_r$.
In the proof of Corollary~\ref{cor:bs-coeff}, we checked that~$[\ul t^{\ul n} \ul z^{\ul 0} ] \tilde R = [1] S_0 S_1^{n_1}\dotsm S_d^{n_d}$, for~$\ul n\in\N^d$, and~$0$ otherwise.
Therefore, $u_{\ul n} = [\ul t^{\ul n} \ul z^{\ul 0} ] R$, for~$n\in\N^d$.
This gives an implementation of Corollaries~\ref{cor:bs-coeff} and~\ref{coro:bs-res}
in the generalized setting.
Algorithm~\ref{algo:sumtores} sums up the procedure for computing integral representations.

\subsection{Analytic integral representations}
\label{sec:anintrep}

When~$\K$ is a subfield of~$\C$, then formal residues can be written as integrals.
\begin{proposition}
  Let~$R(t_{1\smile d}, z_{1\smile e})$ be a rational function whose
  denominator does not vanish when~$t_{1}=\dotsb=t_d =0$. There exist positive
  real numbers~$s_{1\smile d}$ and~$r_{1\smile e}$ such that
  on the set~$\left\{ (t_{1\smile d})\in \C^d \st \forall i, \abs{t_i}\leq s_i \right\}$,
  the power series~$\res_{z_{1\smile e}} R\in \C\llbracket t_{1\smile d}\rrbracket$
  converges and
  \[ \res_{z_{1\smile e}} R = \frac{1}{(2\pi i)^e} \oint_\gamma R(t_{1\smile d}, z_{1\smile e}) \ud z_{1\smile e}, \]
  where~$\gamma = \left\{ z\in\C^e \st \forall 1\leq i \leq e,\ |z_i| = r_i \right\}$.
\end{proposition}

\begin{proof}
  When~$R$ is a Laurent monomial, the 
  equality follows from Cauchy's integral formula.
  By linearity, it still holds when~$R$ is a Laurent polynomial.

  In the general case, let~$R$ be written as~$a/f$, where~$a$ and~$f$ are polynomials.  We may
  assume that the leading coefficient of~$f$ is~$1$ and so~$f$ decomposes
  as~$\mop{lm}(f)( 1 - g)$ where~$g$ is a Laurent polynomial with
  monomials~$\prec 1$.  The hypothesis that~$f$ does not vanish
  when~$t_{1}=\dotsb=t_d =0$ implies that~$\mop{lm}(f)$ depends only on the
  variables~$z_{1\smile e}$ and that~$g$ contains no negative power of the
  variables~$t_{1\smile d}$.
  This and the fact that all monomials of~$g$ are~$\prec 1$ imply that
  there exist positive real numbers~$s_{1\smile d}$ and~$r_{1\smile e}$
  such that~$|g(t_{1\smile d}, z_{1\smile e})| \leq \frac12$
  if~$|t_i| \leq s_i$ and~$|z_i| = r_i$.
  For example, we can take~$s_i = \exp(-\exp(N/i))$ and~$r_i = \exp(-\exp(N/(d+i)))$,
  for some large enough~$N$, because
  \[ \exp(-\exp(N/i))^p = o\left( \exp(-\exp(N/j))^q \right), \quad N\to\infty, \]
  for any~$p,q >0$ and~$i < j$.
  On the one hand
  \[ \res_{z_{1\smile e}} R = \sum_{k\geq 0} \res_{z_{1\smile e}}\left( \frac{a g^k}{\mop{lm}(f)} \right), \]
  by Lemma~\ref{lem:geomsum}, and on the other hand, if~$|t_i|\leq s_i$,
  for~$1\leq i\leq d$ then
  \[ \oint_\gamma R(t_{1\smile d}, z_{1\smile e}) \ud z_{1\smile e} = \sum_{k\geq 0} \oint_\gamma\frac{a g^k}{\mop{lm}(f)} \ud z_{1\smile e}. \]
  where~$\gamma = \left\{ z\in\C^e \st \forall 1\leq i \leq e,\ |z_i| = r_i \right\}$, because the sum~$\sum_{k\geq 0} g^k$
  converges uniformly on~$\gamma$, since~$|g| \leq \frac12$.
  And the lemma follows from the case where~$R$ is a Laurent polynomial.
\end{proof}

\section{Diagonals}
\label{sec:diagonals}

Let~$R(z_{1\smile d}) = \sum_{\ul n\in\N^d} a_{\ul n}\ul z^{\ul n}$ be a rational power series in $\K(z_1, \ldots, z_d)$.
The \emph{diagonal} of~$R$ is the univariate power series
\[ \diag R \eqdef \sum_{n\geq 0} a_{n,\dotsc,n} t^n. \]
Diagonals have be introduced to study properties of the Hadamard product of power series \parencite{CamMar38,Furstenberg1967}. They also appear in the theory of G-functions \parencite{Christol1988}.
They can be written as residues: with~$t\prec z_{2\smile d}$, it is easy to check that
\begin{equation}
  \diag R = \res_{z_2,\dotsc,z_d} \frac{1}{z_2\dotsm z_d} R\left( \frac{t}{z_2\dotsm z_d}, z_2,\dotsc,z_d \right).
\label{eqn:resdiag}
\end{equation}
Despite their simplistic appearance, diagonals have very strong properties, the first of which is \emph{differential finiteness}:
\begin{theorem}
  [\Cite{Chr85}, \cite{Lip88}]
  Let~$R(z_{1\smile d})$ be a rational power series in $\K(z_1, \ldots, z_d)$.
  There exist polynomials~$p_{0\smile r} \in \K[t]$, not all zero, such that~$p_r f^{(r)} + \dotsb + p_1 f' + p_0 f = 0$, where~$f=\diag R$.
  \label{thm:diag-dfinite}
\end{theorem}

We recall that a power series~$f\in \K\llbracket t\rrbracket$ is called \emph{algebraic}
if there exists a nonzero polynomial~$P\in \K[x,y]$ such that~$P(t,f) = 0$.

\begin{theorem}
  [\cite{Furstenberg1967}]
  A power series~$f\in\K\llbracket t\rrbracket$
  is algebraic if and only if~$f$ is the diagonal of
  a \emph{bivariate} rational power series.
  \label{thm:furst1}
\end{theorem}

Diagonals of rational power series in more than two variables need not be algebraic, as shown by
\[ \sum_{n\geq 0} \frac{(3n)!}{n!^3} t^n = \diag\left( \frac{1}{1-z_1-z_2-z_3} \right). \]
However, the situation is simpler modulo any prime number~$p$.
The following result is stated by \textcite[Theorem 1]{Furstenberg1967} over a field of finite characteristic; since reduction modulo~$p$ and diagonal extraction commute, we obtain the following statement:

\begin{theorem}
  [\cite{Furstenberg1967}]
  Let~$f \in \Q\llbracket t\rrbracket$ be the diagonal of a rational power series with rational coefficients.
  Finitely many primes may divide the denominator of a coefficient of~$f$. For all prime~$p$ except those,
  the power series~$f \pmod p \in \F_p\llbracket t\rrbracket$ is algebraic.
  \label{thm:furst2}
\end{theorem}

\begin{example}
It is easy to check that the series~$f=\sum_{n\geq 0} \frac{(3n)!}{n!^3} t^n$
satisfies
\begin{align*}
f &\equiv \left( 1+t \right)^{-\frac{1}{4}} \mod 5, \\
f &\equiv \left( 1+6t+6t^2 \right)^{-\frac 16} \mod 7, \\
f &\equiv \left( 1+6t+2t^2+8t^3 \right)^{-\frac{1}{10}} \mod{11}, \text{ etc.}
\end{align*}
\end{example}

The characterization of diagonals of rational power series remains largely an
open problem, despite very recent attempts \parencite{Christol15}.  A natural
measure of the complexity of a power series~$f\in \Q\llbracket t\rrbracket$ is
the least integer~$n$, if any, such that~$f$ is the diagonal of a rational
power series in~$n$ variables.  Let~$N(f)$ be this number, or~$N(f) = \infty$
if there is no such~$n$.\footnote{This notion is closely related to the
\emph{grade} of a power series \parencite{AMF11}.} It is clear that~$N(f) =
1$ if and only if~$f$ is rational.  According to Theorem~\ref{thm:furst1},
$N(f)=2$ if and only if~$f$ is algebraic and not rational.  It is easy to
find power series~$f$ with~$N(f)=3$, for example~$N\big(\sum_n
\frac{(3n)!}{n!^3} t^n\big)=3$ because it is the diagonal of a trivariate
rational power series but it is not algebraic.  However, we do not know of
any power series~$f$ such that~$N(f) > 3$.

\begin{conjecture}
  [\cite{Christol1990}]
  Let~$f\in\Z\llbracket t\rrbracket$ be a power series with integer
  coefficients, positive radius of convergence and such that~$p_r f^{(r)} +
  \dotsb + p_1 f' + p_0 f = 0$ for some polynomials~$p_{0\smile
  r}\in\Q[t]$, not all zero.  Then~$f$ is the diagonal of a rational power series.
\end{conjecture}

In this section, we prove the following equivalence:
\begin{theorem}\label{thm:main}
  A sequence~$u:\N\to\K$ is a binomial sum if and only if the generating
  function~$\sum_{n\geq 0} u_n t^n$ is the diagonal of a rational power series. 
  \label{thm:equiv-bs-diag}
\end{theorem}
In~\S\ref{sec:binom2diag}, we prove that the generating function of a binomial
sum is a diagonal and in~\S\ref{sec:diag2binom}, we prove the converse.  When
writing that a sequence~$u:\N\to\K$ is a binomial sum, we mean that there
exists a binomial sum~$v:\Z\to\K$ whose support in contained in~$\N$ and
which coincides with~$u$ on~$\N$.  Theorem~\ref{thm:main} can be stated
equivalently without restriction on the support: a sequence~$u:\Z\to\K$ is a
binomial sum if and only if the generating functions~$\sum_{n\geq 0} u_n t^n$
and~$\sum_{n\geq 0} u_{-n} t^n$ are diagonals of rational power series.

Note that \textcite[Theorem 1.3]{GaPa14} proved a similar, although essentially different, result: the subclass of binomial multisums they consider corresponds to the subclass of diagonals of $\mathbb{N}$-rational
power series.

Theorem~\ref{thm:main} and the theory of diagonals of rational functions imply right away interesting corollaries.

\begin{corollary}\label{cor:bs-dfinite}
  If~$u:\Z\to\K$ is a binomial sum, then there exist polynomials~$p_{0\smile
  r}$ in $\K[t]$, not all zero, such that~$p_r(n) u_{n+r} + \dotsb + p_1(n) u_{n+1} +
  p_0(n) u_n = 0$ for all~$n\in\Z$.
\end{corollary}
As mentioned in the introduction, this is a special case of a more general result for proper hypergeometric sums that can alternately be obtained as a consequence of the results of \textcite{Lipshitz1989,AbramovPetkovsek2002}.
\begin{corollary}
  [\cite{FlaSor98}]
  If the generating function~$\sum_n u_nt^n$ of a sequence~$u:\N\to\K$
  is algebraic, then~$u$ is a binomial sum.
\end{corollary}

\begin{corollary}
  Let~$u:\N\to\Q$ be a binomial sum.  Finitely many primes may divide the
  denominators of values of~$u$. For all primes~$p$ except those, the generating
  function of a binomial sum is algebraic modulo~$p$.
\end{corollary}

Moreover, Christol's conjecture is equivalent to the following:
\begin{conjecture}
  If an integer sequence~$u:\N\to\Z$ grows at most exponentially and satisfies a
  recurrence~$p_r(n) u_{n+r} + \dotsb + p_0(n) u_n = 0$, for some
  polynomials~$p_{0\smile r}$ with integer coefficients, not all zero, then~$u$ is a binomial sum.
\end{conjecture}

The proof of Theorem~\ref{thm:main} also gives information on binomial sums depending on
several indices, in the form of a converse of Corollary~\ref{cor:bs-coeff}.
\begin{proposition}\label{converse-coro}
A sequence~$\N^d\to\K$ is a binomial sum if and only if there exists
a rational function~$R(t_{1\smile d}, z_{1\smile r})$, with~$t_{1\smile d}\prec
z_{1\smile r}$, such that~$u_{\ul n}=[\ul t^{\ul n} \ul z^{\ul 0}] R$ for
all~$\ul n\in \Z^d$.
\end{proposition}

\begin{proof}
  [Sketch of the proof]
  The “only if” part is Corollary~\ref{cor:bs-coeff}.
  Conversely, let~$R(t_{1\smile d}, z_{1\smile r})$ be a rational function
  and let~$u_{\ul n}=[\ul t^{\ul n} \ul z^{\ul 0}] R$, for~$\ul n\in\Z^d$.
  We assume that~$u_{\ul n} = 0$ if~$\ul n\notin \N^d$.
  Using the same technique as in~\S\ref{sec:binom2diag}, we show that
  there exist a rational power series~$S(z_{1\smile d+r})$ and monomials~$w_{1\smile d}$
  in the variables~$z_{1\smile d+r}$
  such that~$u_{\ul n} = [\ul w^{\ul n}] S$.
  Then, with the same technique as in~\S\ref{sec:diag2binom}, we write~$u_{\ul n}$ as a binomial sum.
\end{proof}

In other words, binomial sums are exactly the constant terms of rational power series,
where the largest variables, for the order~$\prec$, are eliminated.

\subsection{Binomial sums as diagonals}
\label{sec:binom2diag}

Corollary~\ref{coro:bs-res} and Equation~\eqref{eqn:resdiag} provide an
expression of the generating function of a binomial sum (of one free variable)
as the diagonal of a rational \emph{Laurent} series, but
more is needed to obtain it as the diagonal of a rational \emph{power} series
and obtain the “only if” part of Theorem~\ref{thm:equiv-bs-diag}.

Let~$u:\N\to\K$ be a binomial sum.  In this section, we aim at constructing a
rational power series~$S$ such that~$\sum_{n\geq 0}u_nt^n = \diag S$.  By
Corollary~\ref{cor:bs-coeff}, there exists a rational function~$\tilde R(z_{0\smile
r}) = \frac af$ such that~$u_n = [z_0^n z_1^0 \dotsm z_r^0] \tilde R$ for all~$n\in \Z$.

%For~$i<j$, let~$\eta_{ij} : \Z^d\to\Z^d$ be the increasing endomorphism defined by
%\[ \eta_{ij}(n_{1\smile d}) \eqdef (n_1,\dotsc,\underbrace{n_i}_{\text{$i$\textsuperscript{th}\ pos.}}, \dotsc, \underbrace{n_j + n_i}_{\text{$j$\textsuperscript{th}\ pos.}},\dotsc, n_d).  \]
Recall the notation~$f^\varphi$, for~$f\in\L_{r+1}$ and~$\varphi$ an increasing
endomorphism of~$\Z^{r+1}$, introduced in~\S\ref{sec:laurent}.  For example,
if~$f=z_1+z_2\in\L_2$ and~$\varphi(n_1,n_2) = (n_1,n_1+n_2)$ then~$f^\varphi =
z_1 z_2 + z_2 = z_2(1+z_1)$.

\begin{lemma} For any polynomial $f\in \K[z_{0\smile r}]$ 
  there exists an increasing endomorphism~$\varphi$ of~$\Z^{r+1}$ such
  that~$f^\varphi = C\ul z^{\ul m}(1+g)$, for some~$\ul m\in \N^{r+1}$, $C\in
  \K \setminus \left\{ 0 \right\}$ and $g\in\K[z_{0\smile r}]$ with~$g(0,\dotsc,0)=0$.
  %and~$\varphi$ is the composition of some~$\eta_{ij}$
  \label{lem:monord}
\end{lemma}

\begin{proof}
  Let~$\ul z^{\ul a}$ and~$\ul z^{\ul b}$ be monomials of~$f$ such that~$\ul z^{\ul a}
  \prec \ul z^{\ul b}$. Let~$i$ be the smallest integer such that~$a_i \neq
  b_i$. By definition~$a_i > b_i$.
  For~$k\in\N$, let~$\varphi_k : \Z^{r+1} \to \Z^{r+1}$ be defined by
  \[ \varphi_k : (n_0,\dotsc,n_r) \in \Z^{r+1} \mapsto (n_0,\dotsc,n_i,n_{i+1}+k n_i,\dotsc,n_r + k n_i) \in \Z^{r+1}. \]
  It is strictly increasing and if~$k$ is large enough then $\varphi_k( \ul a )
  \geq \varphi_k(\ul b)$ componentwise, that is~$(\ul z ^{\ul b})^{\varphi_k}$
  divides~$(\ul z^{\ul a})^{\varphi_k}$.
  We may apply repeatedly this argument to construct an increasing
  endomorphism~$\varphi$ of~$\Z^{r+1}$ such that the leading monomial
  of~$f^\varphi$ divides all the monomials of~$f^\varphi$, which proves the
  Lemma.
\end{proof}

Let~$\varphi$ be the endomorphism given by the lemma above applied to the denominator of~$\tilde R$.  For~$0\leq 0 \leq
r$, let~$w_i=z_i^\varphi$.  Then
\[ \tilde R(w_{0\smile r}) = \tilde R^\varphi = \frac{a^\varphi}{C \ul z^{\ul m} (1+g)}, \]
because~$\cdot^\varphi$ is a field morphism, as explained in~\S\ref{sec:laurent}.  And by definition of~$\tilde R$, we
have~$u_n = [w_0^n] \tilde R(w_{0\smile r})$ for all~$n\in\Z$.  Let~$R$ be the
rational power series~$\frac{a^\varphi}{C (1+g)}$, so that~$u_n =
[w_0^n](R/{\ul z}^{\ul m})$.  We now prove that we can reduce to the case
where~$\ul m = 0$.  If~$\ul m\neq 0$, let~$i$ be the smallest integer such that~$m_i\neq
0$.  The specialization~$R|_{z_i=0}$ is a rational power series and~$R -
R_{|z_i=0} = z_i T$ for some rational power series~$T$.  For all~$n\geq 0$, the
coefficient of~$w_0^n$ in~$R|_{z_i=0}/{\ul z^{\ul m}}$ is zero %\todo{recall discussion in \S2.3 for def}
because the
exponent of~$z_i$ in~$w_0^n$ is nonnegative
while the exponent of~$z_i$ in every monomial in~$R|_{z_i=0}/{\ul z^{\ul m}}$ is~$-m_i<0$.  Thus
$u_n = [w_0^n] \frac{T}{\ul z^{\ul m}/z_i}$,
and we can replace~$R$ by~$T$ and subtract~$1$ to the first nonzero
coordinate of~$\ul m$, which makes~$\ul m$ decrease for the lexicographic
ordering.  Iterating this procedure leads to~$\ul m=0$ and thus~$u_n =
[w_0^n]T$ for some rational power series~$T$.

Let us write~$w_0$ as~$z_0^{a_0}\dotsm z_r^{a_r}$, with~$a_{0\smile r}\in\N^{r+1}$.
If all the~$a_i$'s were equal to~$1$, then~$\sum_n u_n t^n$ would be the
diagonal of~$T$. 
First, we reduce to the case where all the~$a_i$'s are positive.
If some~$a_i$ is zero, then
\[ [w_0^n] T = [z_0^{a_0}\dotsm z_{i-1}^{a_{i-1}}z_{i+1}^{a_{i+1}}  \dotsm z_r^{a_r}] T|_{z_i=0} \]
and we can simply remove the variable~$z_i$.

Second, we reduce further to the case where all the~$a_i$'s are~$1$. Let us consider the following rational power series:
\[ U = \frac{1}{a_0\dotsm a_r}\sum_{\varepsilon_0^{a_0} = 1} \dotsb \sum_{\varepsilon_r^{a_r} = 1} T(\varepsilon_0z_0,\dotsc,\varepsilon_r z_r), \]
where the~$\varepsilon_i$ ranges over the~$a_i$-{th} roots of unity.
By construction, if~$m$ is a monomial in the~$z_i^{a_i}$, then~$[m] U = [m] T$.
In particular~$u_n = [w_0^n]U$.

We may consider~$T$ and~$U$ as elements of the extension of the field
$\K(z_0^{a_0},\dotsc,z_r^{a_r})$ by the roots of the polynomials~$X^{a_i} -
z_i$, for~$0\leq i\leq r$.  By construction, the rational function~$U$ is left
invariant by the automorphisms of this extension.
Thus~$U\in\K(z_0^{a_0},\dotsc,z_r^{a_r})$.  Let~$S(z_{0\smile r})$ be the
unique rational function such that~$U=S(z_0^{a_0},\dotsc,z_r^{a_r})$.  It is a
rational power series and~$u_n = [z_0^n\dotsm z_r^n] S$.  Thus~$\sum_{n\geq
0}u_nt^n = \diag S$, which concludes the proof that binomial sums are diagonals of rational functions, the “only if” part of Theorem~\ref{thm:main}.

\subsection{Summation over a polyhedron}
\label{sec:sumpoly}

To prove that, conversely, diagonals of rational power series are generating functions of
binomial sums, we prove a property of closure of binomial sums under certain
summations that generalize the indefinite summation of rule~\ref{def:it:sum} of Definition~\ref{def:bs}.
Let~$u:\Z^{d+e}\to \K$ be a binomial sum and let~$\Gamma\subset \R^{d+e}$ be a rational polyhedron, that is to say
\[ \Gamma = \bigcap_{\lambda\in \Lambda} \left\{ x\in\R^{d+e} \st \lambda(x) \geq 0 \right\}, \]
for a finite set~$\Lambda$ of linear maps~$\R^{d+e}\to \R$ with integer
coefficients in the canonical bases.
This section is dedicated to the proof of the following:
\begin{proposition}\label{prop:somme-poly}
  If for all~$\ul n\in \Z^d$ the set~$\left\{ \ul m\in \R^e \st (\ul n, \ul m)\in \Gamma \right\}$ is bounded, then the sequence
  \[ v : \ul n\in \Z^d \mapsto \sum_{\ul m\in \Z^e} u_{\ul n,\ul m} \mathbf1_{\Gamma}(\ul n,\ul m) \]
  is a binomial sum, where~$\mathbf1_{\Gamma}(\ul n,\ul m) = 1$ if~$(\ul n,\ul m)\in\Gamma$ and~$0$ otherwise.
\end{proposition}

Recall that~$H:\Z\to\K$ is the sequence defined by~$H_n = 1$ if~$n\geq 0$ and~$0$ otherwise; it is a binomial sum, see~\S\ref{sec:bsexamples}.
Thus~$\mathbf1_\Gamma$ is a binomial sum, because~$\mathbf1_\Gamma(\ul n) = \prod_{\lambda\in\Lambda} H_{\lambda(\ul n)}$ and each factor of the product is a binomial sum, thanks to rule~\ref{def:it:lambda} of the definition of binomial sums.
However, the summation defining~$v$ ranges over an infinite set, rule~\ref{def:it:sum} is not enough to conclude directly, neither is the generalized summation of Equation~\eqref{eqn:generalizedsummation}.

For~$\ul m\in\R^e$, let~$\abs{\ul m}_\infty$ denote~$\max_{1\leq i\leq e} |m_i|$.
For~$\ul n\in \Z^d$, let~$B(\ul n)$ be
\[ B(\ul n) \eqdef \sup\left\{ \abs{\ul m}_\infty \st \ul m \in \R^e \text{ and } (\ul n, \ul m) \in \Gamma \right\} \cup \left\{ -\infty \right\}. \]
By hypothesis on~$\Gamma$, $B(\ul n) < \infty$ for all~$\ul n\in\Z^d$.

\begin{lemma}
  There exists~$C>0$ such that~$B(\ul n)\leq C(1+\abs{n}_\infty)$, for all~$\ul n\in \Z^d$.
  \label{lem:linbound}
\end{lemma}

\begin{proof}
  By contradiction, let us assume that such a~$C$ does not exist, that is there
  exists a sequence~$\ul p_k = (\ul n_k,\ul m_k)$ of elements of~$\Gamma$ such
  that~$\abs{\ul m_k}_\infty/\abs{\ul n_k}_\infty$ and~$\abs{\ul n_k}_\infty$
  tend to~$\infty$.
  Up to considering a subsequence, we may assume that~$\ul m_k/\abs{\ul m_k}_\infty$
  has a limit~$\ell\in\R^e$, which is nonzero.
  For all~$\alpha\in[0,1]$ and~$k\geq 0$, the point~$\ul p_0 + \alpha (\ul p_k - \ul p_0)$ is in~$\Gamma$, because~$\Gamma$ is convex.
  Let~$u>0$ and let~$\alpha_k = u / \abs{\ul m_k}_\infty$. If~$k$ is large enough, then~$0\leq \alpha_k \leq 1$. Moreover
  \[ \ul p_0 + \alpha_k (\ul p_k - \ul p_0) \mathop{\to}\limits_{k\to\infty} (\ul n_0, \ul m_0 + u \ell). \]
  Yet~$\Gamma$ is closed so~$(\ul n_0, \ul m_0 + u \ell) \in \Gamma$.
  By definition~$ \abs{\ul m_0 + u \ell}_\infty \leq B(\ul n_0) < \infty$.
  This is a contradiction because~$\ell \neq 0$ and~$u$ is arbitrarily large.
\end{proof}
 
Thus, for all~$\ul n\in\Z^d$,
\begin{equation}
  v_{\ul n} = \sum_{\substack{\ul m\in \Z^e \\ |\ul m| \leq C (1+|\ul n|_\infty)}} u_{\ul n,\ul m} \mathbf1_{\Gamma}(\ul n,\ul m).
  \label{eqn:sumwithbounds}
\end{equation}
For~$1\leq i \leq d$, let~$w^{i,+}$ be the sequence
\[ w^{i,+}_{\ul n} \eqdef H_{n_i} \prod_{j=1}^{i-1}( H_{n_i - n_j -1}  H_{n_i + n_j -1})  \prod_{j=i+1}^d (H_{n_i - n_j}H_{n_i + n_j}). \]
After checking that for any~$a,b\in \Z$, $H_{a+b}H_{a-b} = 1$ if~$a \geq \abs{b}$ and~$0$ otherwise, we see
that~$w^{i,+}_{\ul n} = 1$ if~$n_i =
|n|_\infty$ and~$|n_j| < |n_i|$ for all~$j<i$ ; and~$0$ otherwise.  Likewise,
let~$w^{i,-}$ be the sequence
\[ w^{i,-}_{\ul n} \eqdef H_{-n_i-1} \prod_{j=1}^{i-1}( H_{-n_i - n_j -1}  H_{-n_i + n_j -1})  \prod_{j=i+1}^d (H_{-n_i - n_j}H_{-n_i + n_j}), \]
so that~$w^{i,-}_{\ul n}=1$ if~$n_i = -|n|_\infty < 0$ and~$|n_j| < \abs{n_i}$ for all~$j<i$ ; and~$0$ otherwise.
The~$2d$ sequences $w^{i,+}$ and~$w^{i,-}$ are binomial sums
that partition~$1$:
%\todo{Non: le cas où le max est atteint plusieurs fois n'est pas couvert.}
for all~$\ul n\in\Z^d$
\[ 1 = \sum_{i=1}^d w^{i,+}_{\ul n} + \sum_{i=1}^d w^{i,-}_{\ul n}. \]
By design, the sum in Equation~\eqref{eqn:sumwithbounds} rewrites as
\begin{multline*}
   v_{\ul n} = \sum_{i=1}^d w^{i,+}_{\ul n} \sum_{m_1=-C(1+n_i)}^{C(1+n_i)}\dotsm\sum_{m_e=-C(1+n_i)}^{C(1+n_i)} u_{\ul n,\ul m} \mathbf1_{\Gamma}(\ul n,\ul m) \\
   + \sum_{i=1}^d w^{i,-}_{\ul n} \sum_{m_1=-C(1-n_i)}^{C(1-n_i)}\dotsm\sum_{m_e=-C(1-n_i)}^{C(1-n_i)} u_{\ul n,\ul m} \mathbf1_{\Gamma}(\ul n,\ul m),
\end{multline*}
which concludes the proof of Proposition~\ref{prop:somme-poly}, because now the
summation bounds are affine functions with integer coefficients of~$\ul n$.

From a practical point of view, summations over polyhedra can be handled in a very different
way, see~\S\ref{sec:brion-barvinok}.

\subsection{Diagonals as binomial sums}
\label{sec:diag2binom}

We now prove the second part of Theorem~\ref{thm:equiv-bs-diag}: the diagonal
of a rational function is the generating function of a binomial sum.
Let~$R(z_{1\smile d})$ be a rational power series and let us prove that the
sequence defined by~$u_n = [z_1^n\dotsm z_d^n] R$ is a binomial sum.
Since the binomial sums are closed under linear combinations, it is enough to consider the case where the numerator of~$R$ is a monomial and where the constant term of its denominator is~$1$. (It cannot be zero since~$R$ is a power series.)
Thus~$R$ has the form
\[ R = \frac{\ul z^{\ul m_0}}{1 + \sum_{k=1}^e a_k \ul z^{\ul m_k}}, \]
where the~$\ul m_k$'s have nonnegative coordinates and~$\ul m_k \neq 0$ if~$k\neq 0$.
Let~$y_{0\smile e}$ be new variables and let~$S$ be the rational power series
\[ S(y_{0\smile e}) \eqdef \frac{y_0}{1 + \sum_{k=1}^e a_k y_k} = y_0 \sum_{\ul k \in \N^e} \underbrace{\binom{k_1 + \dotsb + k_e}{k_1,\dotsc,k_e} a_1^{k_1}\dotsm a_e^{k_e}}_{C_{\ul k}} y_1^{k_1}\dotsm y_e^{k_e}. \] 
The coefficient sequence~$C_{\ul k}$ of this power series is a binomial sum because the multinomial coefficient is a product of binomial coefficients:
\[ \binom{k_1 + \dotsb + k_e}{k_1,\dotsc,k_e} = \binom{k_1 + \dotsb + k_e}{k_1}\binom{k_2 + \dotsb + k_e}{k_2}\dotsm \binom{k_{e-1} + k_e}{k_{e-1}}. \]
Let~$\Gamma\subset\R\times\R^{e}$ be the rational polyhedron defined by
\[ \Gamma \eqdef \left\{ (n,\ul k)\in \R\times \R^{e} \st
  k_1 \geq 0,\dotsc,k_e \geq 0  \text{ and } \ul m_0 + \smash{\sum_{i=1}^e}\mathrlap{\phantom{\sum}} k_i \ul m_i = (n,\dotsc,n) \right\}.
\]
By construction~$R(z_{1\smile d}) = S(\ul z^{\ul m_0},\dotsc,\ul z^{\ul m_e})$,
and the image of a monomial~$y_0 y_1^{k_1}\dotsm y_e^{k_e}$ is a diagonal
monomial~$z_1^n\dotsm z_d^n$ if and only if~$(n,\ul k)\in \Gamma$.  Thus
\[ [z_1^n \dotsm z_d^n] R = \sum_{k\in\Z^e} C_{\ul k} \mathbf1_\Gamma(n,\ul k). \]
Thanks to the positivity conditions on the~$\ul m_k$'s, it is obvious that~$\Gamma$
satisfies the finiteness hypothesis of Proposition~\ref{prop:somme-poly}.
Thus the sequence~$u_n$ is a binomial sum.

\section{Computing binomial sums}
\label{sec:computing}

\emph{Computing} may have different meanings.  When manipulating sequences like
binomial sums, one may want, for example, to compute recurrence relations that they satisfy, to
decide their equality, to compute their asymptotic behavior or to find \emph{simple}
closed-form formulas for them.
%\begin{itemize}
%  \item compute recurrence relations;
%  \item decide equality;
%  \item compute asymptotic beahaviours;
%  \item find \emph{simple} closed form formulas.
%\end{itemize}
The representation of the generating series of binomial sums as residues or
diagonals of rational functions gives an interesting tool to tackle these goals
for binomial sums, though some questions remain open.

While the customary tool to compute recurrence relations satisfied by binomial
sums is \emph{creative telescoping}, integral representations of binomial sums
give another approach: given a binomial sum we first compute an
integral representation (this is fast and easy, see
Algorithm~\ref{algo:sumtoct}), and next we compute a differential equation
satisfied by the integral, which translates immediately into a recurrence relation
for the binomial sum. The \emph{decision problem}~$A=B$ can be solved by computing recurrence relations
for~$A - B$ and checking sufficiently many initial conditions, see~\S\ref{sec:eqtest} for
more details.

A recurrence relation is also a good starting point for deriving the \emph{asymptotic
behavior} of a univariate binomial sum \parencite[e.g.][]{MezzarobbaSalvy2010}.
However, some constants that determine the asymptotic behavior can be
computed numerically but it is not known how to decide their nullity, which
makes it difficult to catch the subdominant behavior.  A more direct method is
possible in many cases once the generating function has been written as the diagonal of a
rational power series \parencite{PemantleWilson2013,SalvyMelczer2016}.

The problem of \emph{simplifying} binomial sums is still largely open: for example,
given the left-hand side of Strehl's identity (see~\S\ref{sec:strehl}), it is
not known how to \emph{discover} automatically the right-hand side. The only
known automatic simplification procedures consist in computing a recurrence
relation and applying the algorithm of \textcite{Pet92} to find
hypergeometric solutions or the algorithm of \textcite{AbramovPetkovsek1994} to find D'Alembertian solutions.
See~\S\ref{sec:known-identities} and~\S\ref{sec:proof-conjectures} for numerous
examples.

The computation of a recurrence relation for a binomial sum is done in two steps. Firstly, an integral representation is found
and then a Picard-Fuchs equation is computed.

\subsection{Picard-Fuchs equations}
\label{sec:pf}

Let~$\L$ be a field of characteristic zero and let~$R(z_{1\smile r})$
be a rational function with coefficients in~$\L$, written as~$R=\frac{P}{F}$ where~$P$
and~$F$ are polynomials.  Let~$A_F$ be the localized ring~$\L[z_{1\smile r},
F^{-1}]$. It is known that the~$\L$-linear quotient space
\[ H_F \eqdef A_F / \left( \tfrac{\partial}{\partial z_1} A_F + \dotsb + \tfrac{\partial}{\partial z_r} A_F \right) \]
is finite-dimensional \parencite{Gro66}.
Let us assume that there is a derivation~$\partial$ defined on~$\L$.  It
extends naturally, with~$\partial z_i = 0$, to a derivation on~$A_F$ that commutes with the
derivations~$\tfrac{\partial}{\partial z_i}$, so that~$\partial$ defines a
derivation on the~$\L$-linear space~$H_F$.
Since~$H_F$ is finite-dimensional, there exist~$c_{0\smile m}\in\L$, not all zero,
such that
\[ c_m \partial^m R + \dotsb + c_1 \partial R + c_0 R \in \tfrac{\partial}{\partial z_1} A_F + \dotsb + \tfrac{\partial}{\partial z_r} A_F. \]
Now let us assume that~$\L$ is the field of rational functions~$\K(t_{1\smile d})$ and that~$\partial$ is the derivation~$\tfrac{\partial}{\partial t_i}$
for some~$i$.
Then the operator~$\partial$ commutes with the operator~$\res_{z_{1\smile r}}$, as do the multiplications by elements of~$\L$,
and Lemma~\ref{lem:resders} implies that
\[ c_m \tfrac{\partial^m}{\partial t_i^m} f + \dotsb + c_1 \tfrac{\partial}{\partial t_i} f + c_0 f = 0, \]
where~$f = \res_{z_{1\smile r}} R$, with~$t_{1\smile d} \prec z_{1\smile r}$.
Differential equations that arise in this way are called \emph{Picard-Fuchs equations}.

Recall that~$\L_d$ is the field of iterated Laurent series introduced in~\S\ref{sec:laurent}.
A series~$f(t_{1\smile d})\in\L_d$ is called \emph{differentially
finite} if the~$\K(t_{1\smile d})$-linear subspace of~$\L_d$ generated by the
derivatives ${\partial^{n_1+\dotsb+n_d} f}/{\partial t_1^{n_1}\dotsm
\partial t_d^{n_d}}$, for~$n_{1\smile d} \geq 0$, is finite-dimensional. In
particular, a univariate Laurent series~$f\in\L_1$ is differentially finite if
and only if there exist~$m\geq 0$ and polynomials~$p_{0\smile m}\in
\K[t]$, not all zero, such that~$p_m f^{(m)} + \dotsb + p_1 f' + p_0 f = 0$.
In that case, we say that~$f$ is solution of a differential equation of
\emph{order}~$m$ and \emph{degree}~$\max_k \deg p_k$.
The above argument implies the following classical theorem of which \textcite{Lip88} gave an elementary proof. 

\begin{theorem}
  For any rational function~$R(t_{1\smile d}, z_{1\smile r})$ with~$t_{1\smile d}\prec
  z_{1\smile r}$, the residue~$\res_{z_{1\smile r}} R$ is differentially finite.
  \label{thm:res-dfinite}
\end{theorem}

In previous work \parencites[Theorem~12]{BostanLairezSalvy2013}[\S
I.35.3]{LaiDoct}, we proved the following quantitative result about the size
of Picard-Fuchs equations and the complexity of their computation.  We also described an
efficient algorithm to compute Picard-Fuchs equations \parencite{Lai15}.
\begin{theorem}
  Let~$R \in \K(t, z_{1\smile r})$ be a rational function,
  written as~$R = \frac PF$, with~$P$ and~$F$
  polynomials.
  Let
  \[ N = \max( \deg_{z_{1\smile r}} P + r + 1, \deg_{z_{1\smile r}}F) \ \text{and}\ d_t = \max(\deg_t P, \deg_t F). \]
  Then~$\res_{z_{1\smile r}} R$, with~$t\prec z_{1\smile r}$, is solution of a
  linear differential equation of order at most~$N^r$ and degree at
  most~$(\frac 58 N^{3r} + N^r) \exp(r) d_t$. Moreover, this differential equation
  can be computed with~$\mathcal O(\exp(5r) N^{8 r} d_t)$ arithmetic operations
  in~$\K$, uniformly in all the parameters.
  \label{thm:pf-bounds}
\end{theorem}

\subsection{Power series solutions of differential equations}
\label{sec:diffeqseries}

When a power series satisfies a given linear differential equation with polynomial
coefficients, one only needs a few {initial conditions} to determine
entirely the power series.

Let~$\cL \in \K[t]\langle\partial_t\rangle$ be a linear differential operator
in~$\partial_t$ with polynomial coefficients in~$t$. There exists a
unique~$n\in\Z$ and a unique~$b_{\cL}\in \K[a]$ such that~$\cL(t^a) =
b_{\cL}(a) t^{a+n} + o(t^{a+n})$ for all~$a\in\Z$ and~$t\to 0$.  The polynomial~$b_{\cL}$ is
the \emph{indicial polynomial} of~$\cL$ at~$t=0$.  For more
details about the indicial polynomial, see \citetitle{Inc44} \parencite{Inc44}.  It is easy to check that
if~$f\in\K\llbracket t\rrbracket$ is annihilated by~$\cL$, then its leading
monomial~$t^n$ satisfies~$b_\cL(n)=0$.
\begin{proposition}\label{prop:indeq}
  Let~$f\in\K\llbracket t\rrbracket$. If~$\cL(f) = 0$
  and if~$[t^n]f = 0$ for all~$n\in \N$ such that~$b_{\cL}(n) = 0$, then~$f=0$.
\end{proposition}

It is worth noting that the indicial
equation of a Picard-Fuchs equation is very special:
\begin{theorem}[\cite{Kat70}]
  If~$\cL$ is a Picard-Fuchs equation, then the degree of~$b_\cL$
  equals the order of~$\cL$ and all the roots of~$b_\cL$ are rational.
\end{theorem}

The data of a differential operator~$\cL$ and elements of~$\K$ for each
nonnegative integer root of~$b_\cL$ (that we will call here \emph{sufficient
initial conditions}) determines entirely an element of~$\K\llbracket
t\rrbracket$.  It is an excellent data structure for manipulating power
series \parencite{SalvyZimmermann1994}.  For example, it lets one compute
efficiently the coefficients of the power series~$\sum_n u_n t^n$:
the differential equation translates into a recurrence relation
\[ p_r(n) u_{n+r} + \dotsc + p_1(n) u_{n+1} + p_0(n) u_n = 0 \]
for some polynomials~$p_{0\smile r}\in\K[n]$ and the
sufficient initial conditions translate into initial conditions for the
recurrence, exactly where we need them.

\subsection{Equality test for univariate differentially finite power series}
\label{sec:eqDfin}
Let~$f\in\K\llbracket t\rrbracket$ be a power series given by a differential
operator~$\cL\in\K[t]\langle \partial_t \rangle$ such that $\cL(f) = 0$, and by sufficient initial
conditions.  Let~$\cM$ be another differential operator. We may decide
if~$\cM(f)=0$ in the following way.  Firstly, we compute the right g.c.d.
of~$\cM$ and~$\cL$: this is the operator~$\cD$ of the largest order such that~$\cM
= \cM' \cD$ and~$\cL = \cL' \cD$ for some operators~$\cM'$ and~$\cL'$
in~$\K(t)\langle \partial_t \rangle$.  Then, it is enough to compute the
indicial equation~$b_{\cL'}$ and  to compute~$[t^n] \cD(f)$ for each
nonnegative integer root~$n$ of~$b_{\cL'}$.  We will find only zeros if and
only if~$\cM(f) = 0$.  Indeed, $\cM(f) = 0$ if and only if~$\cD(f)= 0$ and
since~$\cL'(\cD(f)) = 0$, we may apply Proposition~\ref{prop:indeq} to check
whether~$\cD(f)= 0$ or not.

Now, let~$g\in\K\llbracket t\rrbracket$ be another power series given by a
differential operator~$\cM\in\K[t]\langle \partial_t \rangle$ and sufficient
initial conditions. To decide if~$f=g$, it is enough to check that~$\cM(f)=0$,
with the above method, and to check that the coefficients of~$f$ and~$g$
corresponding to the nonnegative integer roots of~$b_\cM$ are the same.

\subsection{Equality test for binomial sums}
\label{sec:eqtest}

Let~$u : \Z\to\K$ be a binomial
sum.  Up to considering separately the binomial sums~$H_n u_n$ and~$H_n
u_{-n}$, it is enough to look at the case where~$u_n = 0$ for~$n<0$.  In this
case~$u$ is entirely determined by its generating function~$f=\sum_{n\geq 0}u_n
t^n$.  Using Algorithm~\ref{algo:sumtoct} we obtain an integral representation
of~$f$, and then, as explained in~\S\ref{sec:pf}, we obtain a differential
operator~$\cL$ that annihilates~$f$.  Since~$u$ is a binomial sum given
explicitly, we can compute sufficient initial conditions.
Given another binomial sum~$v:\N\to\K$, we can check that~$u=v$ by computing a
differential operator annihilating the generating function~$g$ of~$v$ together
with sufficient initial conditions and by checking that~$f=g$ with the method
explained in~\S\ref{sec:eqDfin}.

Let us now consider the multivariate case.  As above, we can reduce the
equality test for binomial sums~$\Z^d\to\K$ to the equality test for binomial
sums~$\N^d\to\K$. And this equality test can be reduced to the equality test
for binomial sums~$\N^{d-1}\to\K$ as follows.
Let~$u:\N^d\to\K$ be a binomial sum. It is determined by its generating
function~$f\left( t_{1\smile d} \right) = \sum_{\ul n\in \N^d} u_{\ul n} \ul
t^{\ul n} \in \L_d$.  With Algorithm~\ref{algo:sumtoct}, we compute a rational
function~$R(t_{1\smile d},z_{1\smile r})$ such that~$f = \res_{z_{1\smile r}}
f$. Let~$\K'$ be the field~$\K(t_{1\smile d-1})$.  Following~\S\ref{sec:pf}, we
can compute an operator~$\cL\in\K'[t_d]\langle \partial_{t_d}\rangle$ such
that~$\cL(f)=0$.  This gives a differential equation for~$f$ considered as a
power series in~$t_d$ with coefficients in~$\L_{d-1}$.  The sufficient initial
conditions to determine~$f$ are given by the power series~$\sum_{\ul
n\in\N^{d-1}} u_{\ul n, k} \ul t^{\ul n}$, where~$k$ ranges over the
nonnegative integer roots of~$b_\cL$.  These power series are the generating
functions of binomial sums in~$d-1$ variables.  This reduces the equality test
for binomial sums in~$d$ variables to the equality test for binomial sums
in~$d-1$ variables.

\section{Geometric reduction}
\label{sec:geomred}

Putting into practice the computation of binomial sums through integral
representations shows immediately that the number of integration variables is
high and makes the computation of the Picard-Fuchs equations slow. However, the
rational functions obtained this way are very peculiar. For example, their
denominator often factors into small pieces.  This section presents a
sufficient condition under which a residue~$\res_v F$ of a rational
function~$F$ is rational. This leads to rewriting an iterated residue of a
rational function, like the ones given by Corollary~\ref{coro:bs-res}, into
another one with one or several variables less.  This simplification procedure
is very efficient on the residues we are interested in and reduces the number
of variables significantly. 

Conceptually the simplification procedure is simple: in terms of integrals it
boils down to computing partial integrals in specific cases where we know that
they are rational. The rational nature of a partial integral depends on the
integration cycle and integration algorithms usually forget about this cycle.
Instead, they compute the Picard-Fuchs equations---see~\S\ref{sec:pf}---that 
annihilate a given integral for \emph{any} integration cycle.  In our
setting, the integration cycle underlies the notion of residue---see~\S\ref{sec:anintrep}.
This provides a symbolic treatment that we call
\emph{geometric reduction} since it takes into account the geometry of the
cycle and decreases the number of variables for which the general integration
algorithm above is actually needed.  The time required by the computation is
dramatically reduced by this symbolic treatment.

\subsection{Rational poles}
\label{sec:ratpoles}

Let us consider variables~$v$ and~$z_{1\smile d}$, where~$v$ can appear anywhere in the variable ordering.
Let~$F(v, z_{1\smile d}) = a/f$ be a rational function.  In general, $\res_v F$
is not a rational function---except if~$v$ is the smallest variable, see~\S\ref{sec:residues}.

Let~$\rho\in \L_d$ be a power series in the variables~$z_{1\smile d}$.  The
\emph{classical} residue of~$F$ at~$v=\rho$, denoted~$\Res_{v=\rho} F$, is the
coefficient of~$1/(v-\rho)$ in the partial fraction decomposition of~$F$ as an
element of~$\L_d(v)$.  Contrary to the \emph{formal} residue~$\res_v F$, the
classical residue is always in the field generated by~$\rho$ and the
coefficients of~$F$.  Similarly to the formal residues, the classical residues
of a derivative~$\partial G/\partial v$ are all zero.

Classical residues can be computed in a simple way: if~$\rho$ is not a
pole of~$F$, then~$\Res_{v=\rho} F = 0$; if~$\rho$ is a pole of order~$1$,
then~$\Res_{v=\rho}F = \left( (v-\rho)F \right)|_{v=\rho}$; and if~$\rho$ is a
pole of order~$r>1$, then
\[ \Res_{v=\rho}F = \frac{1}{(r-1)!}\left.\left( \frac{\partial^{r-1}(v-\rho)^r F}{\partial^{r-1} v} \right)\right|_{v=\rho}. \]

In its simplest form, the geometric reduction applies when~$f$ factors over~$\K(z_{1\smile d})$ as a
product of factors of degree~$1$:
\[ f = C(z_{1\smile d}) \prod_{\rho \in U}(v-\rho)^{n_\rho}, \]
where~$U$ is a finite subset of~$\K(z_{1\smile d})$.
Then the partial fraction decomposition of~$F$ writes as
\[ F = \sum_{\rho\in U}\left( \frac{a_\rho}{v-\rho} + \sum_{k>1}\frac{b_{\rho,k}}{(v-\rho)^k} \right) + P(v), \]
where~$a_\rho\in \K(z_{1\smile d})$ is~$\Res_{v=\rho} F$,
where~$b_{\rho,k}\in\K(z_{1\smile d})$ and where~$P\in \K(z_{1\smile
d})[v]$.
The terms with multiple poles and the polynomial term are derivatives and thus
their formal residue with respect to~$v$ are zero.  Hence
\[ \res_v F = \sum_{\rho\in U} \res_v\left( \frac{a_\rho}{v-\rho} \right). \]
Let~$\rho\in U$. Depending on the leading monomial~$\lm(\rho)$ of~$\rho$, as an element of~$\L_d$, two situations may occur.
Either~$\lm(\rho) \prec v$, in which case
\[ \frac{a_\rho}{v - \rho} = \frac{a_\rho}{v}\sum_{n=0}^{\infty} \left( \frac{\rho}{v} \right)^n, \]
and~$\res_v \frac{a_\rho}{v - \rho} = a_\rho$; or~$\lm(\rho) \succ v$, in which case
\[ \frac{a_\rho}{v - \rho} = -\frac{a_\rho}{\rho}\sum_{n=0}^{\infty} \left( \frac{v}{\rho} \right)^n, \]
hence~$\res_v \frac{a_\rho}{v - \rho} = 0$. Since the variable~$v$ does not appear in~$\rho$, the equality~$\lm(\rho) = v$ cannot happen.
In the end, we obtain that
\begin{equation}
  \res_v F = \sum_{\rho\in U} [\lm(\rho)\prec v] \Res_{v=\rho}F,
  \label{equ:res-simple}
\end{equation}
where the bracket is~$1$ if the inequality inside is true and~$0$ otherwise.
In particular, the right-hand side is a rational function that we can compute.

\begin{example}\label{exem:simpl1}
  Let~$d > 0$ be an integer and let us consider the binomial sum
  \[ u_n =
    \sum_{k=0}^n (-1)^k \binom{n}{k}\binom{dk}{n}. \]
  We will show that~$u_n = (-d)^n$.
  This example is interesting because the geometric reduction procedure alone is
  able to compute entirely the double integral representing the generating function of~$u$,
  whereas Zeilberger's algorithm finds a recurrence relation of order~$d-1$ \parencite{PauleSchorn1995}, far from the minimal recurrence relation~$u_{n+1}+du_n=0$.

  Algorithm~\ref{algo:sumtoct} computes that
  \[ \sum_{n \geq 0}u_n t^n = \res_{  z_1,z_2 } \frac{z_2}{(z_2-t(1+z_1))(z_1z_2 + t(1+z_1)(1+z_2)^d)}, \]
  with~$t\prec z_1 \prec z_2$.
  Let~$F$ denote the rational function on the right-hand side.
  Each factor of the denominator of~$F$ has degree~$1$ with respect to~$z_1$.
  Thus Equation~\eqref{equ:res-simple} applies. The roots of the denominator are
  \[ \rho_1 = \frac{z_2}{t}-1 \quad\text{and}\quad \rho_2 = \frac{-1}{1+\frac{z_2}{t(1+z_2)^d}}, \]
  moreover~$\lm(\rho_1) = z_2/t \succ z_1$ and~$\lm(\rho_2) = t/z_2 \prec z_1$.
  Thus~$\res_{z_1} F = \Res_{z_1=\rho_2} F$ and we obtain
  \[ \sum_{n \geq 0}u_n t^n = \res_{z_2}\res_{z_1} F = \res_{z_2} \frac{1}{t(1+z_2)^d+z_2-t}. \]
  If~$d>2$, the denominator of the latter rational function does not split into
  factors of degree~$1$ and Equation~\eqref{equ:res-simple} does not apply.
  However, the study of nonrational poles can lead to a further reduction. 
\end{example}

\subsection{Arbitrary poles}\label{sec:arbitrary-poles}
Equation~\eqref{equ:res-simple} extends to the general case.
To describe this generalization, we need an algebraic closure of~$\L_d$.
Let~$\L_{d,N}$ be the field
\[ \L_{d,N} \eqdef \K\slf{z_d^{1/N}}\slf{z_{d-1}^{1/N}}\dotsc\slf{z_2^{1/N}}\slf{z_1^{1/N}}. \]
It is the algebraic extension of~$\L_d$ generated by the~$z_i^{1/N}$.
Let~$\L_{d,\infty}$ be the union of all~$\L_{d,N}$, $N>0$, and let~$\overline{\K}\otimes_\K \L_{d,\infty}$ be the
compositum of~$\L_{d,\infty}$ and~$\overline{\K}$, where~$\overline{\K}$ is an algebraic closure of~$\K$.
The following is classical \parencite[e.g.][]{Ray74}.
\begin{lemma}
  [Iterated Puiseux theorem]
  The field~$\overline{\K}\otimes_\K \L_{d,\infty}$ is an algebraic closure of~$\L_d$.
\end{lemma}
The field~$\overline{\K}\otimes_\K \L_{d,\infty}$ is thus simply denoted~$\overline{\L_d}$.
The valuation defined on~$\L_d$ is extended to a valuation defined on~$\overline{\L_d}$
with values in the group~$\Q^d$, ordered lexicographically.
The leading monomial~$\lm(\rho)$ of an~$f\in\overline{\L_d}$ is also defined as~$\ul z^{v(f)}$.
The argument of~\S\ref{sec:ratpoles} applies and shows that
\begin{equation}
  \res_v F = \sum_{\text{$\rho$ pole of~$F$}} [\lm(\rho) \prec v] \Res_{v=\rho}F,  
  \label{equ:algpoles}
\end{equation}
where this time, the poles are in~$\overline{\L_d}$.
A root~$\rho$ is called
\emph{small} if~$\rho=0$ or~$\lm(\rho) \prec v$ and \emph{large} otherwise.

\begin{example}
  Let~$F = \frac{1}{xy(y^2+y-x)}$, with~$x \prec y$.
  With respect to~$y$, the poles of~$F$ are~$0$, and
  \[ \rho_1 = -\frac12 + \frac{\sqrt{1+4x}}{2} \quad\text{and}\quad \rho_2 =  -\frac12 - \frac{\sqrt{1+4x}}{2}. \]
  Only~$0$ and~$\rho_1$ are small.
  Thus
  \begin{align*}
    \res_y F &= \Res_{y=0} F + \Res_{y=\rho_1} F = -\frac{1}{x^2} + \frac{2}{x+4x^2  - x\sqrt{1+4x}}.
  \end{align*}
\end{example}

Equation~\eqref{equ:algpoles} does not look as interesting as
Equation~\eqref{equ:res-simple} because the right-hand side is algebraic and
need not be a rational function.  However, if all roots are large, then the
residue is zero, which is interesting.  On the contrary, if they are all small,
then~$\res_v F$ is the sum of all the classical residues of~$F$, which equals
the residue at infinity:
\[ \res_v F =
  \Res_{v=\infty}F \eqdef \Res_{v=0}\left(
  -\frac{1}{v^2}F |_{v\gets 1/v} \right), \]
which is a rational function.  Thus, in the case where the poles are either all small
or all large, Equation~\eqref{equ:algpoles} shows that the residue~$\res_v F$ is
rational and shows how to compute it.

Actually it is enough to check that any two conjugate poles (two poles are conjugate if they annihilate the same irreducible factor of the denominator of~$F$) are simultaneously large or small.
Indeed, we can write~$f = \prod_{k=1}^r f_k^{n_k}$
where~$f_{1\smile r}$ are irreducible polynomials in~$\K(z_{1\smile d})[v]$,
and the partial fraction decomposition leads to
\[ F = \sum_{k=1}^r \frac{a_k}{f_k^{n_k}} \]
for some polynomials~$a_{1\smile r}$.
Equation~\eqref{equ:algpoles} applies to each term of this sum.
If~$U_k$ denotes the set of all the roots of~$f_k$, then
\[ \res_v F = \sum_{k=1}^r \sum_{\rho \in U_k} [\lm(\rho) \prec v] \Res_{v=\rho}\frac{a_k}{f_k^{n_k}}, \] 
and we can apply the \emph{all large or all small} criterion to each subsum separately.

\begin{proposition}\label{prop:geomred}
  With the notations above, if for all~$k$, there exists~$\varepsilon_{k} \in
  \left\{ 0,1 \right\}$ such that~$[\lm(\rho) \prec v] = \varepsilon_k$ for
  all~$\rho\in U_k$, then
  \[ \res_v F = \sum_{k=1}^r \varepsilon_k \Res_{v=\infty} \frac{a_k}{f_k^{n_k}}.\]
\end{proposition}

It remains to explain how to compute the leading monomials of the roots of a
polynomial~$f\in \K(z_{1\smile d})[v]$. If~$d=1$, then it is the classical
method of Newton's polygon for the resolution of bivariate polynomial equations
using Puiseux series \parencite[e.g.][]{Wal50,Cut04}.  When~$d\geq 2$,
we may apply this method recursively by considering~$\overline{\L_d}$ as a
subfield of the field of Puiseux series with coefficients
in~$\overline{\L_{d-1}}$.  Based on this idea, Algorithm~\ref{algo:tor}
proposes an implementation which avoids all technicalities.
Algorithm~\ref{algo:ratres} sums up the geometric reduction procedure with
respect to the variable~$v$. Of course one should apply the procedure
iteratively with every variable until no further reduction is possible.
The success of the reduction may depend on the order in which we eliminate all the variables.
It is possible to try every possible order efficiently thanks to the following compatibility property: if both $\textsc{GeomRed}(\textsc{GeomRed}(R, u), v)$ and $\textsc{GeomRed}(\textsc{GeomRed}(R, v), u)$ do not fail, then they are both equal to~$\res_{u,v}(R)$.

\begin{algo}[t]
\centering
\begin{algorithmic}
  \Function{\allbigallsmall}{$S, k$}
    \If{$\max_{m\in S} m_k = \min_{m\in S} m_k$}
        \Return $\varnothing$
    \EndIf
    \If{$k = 1$}
      \Return $\{\textsc{out}\}$
    \EndIf
    \State $\mu\gets\min_{m\in S} m_1$
    \State $M\gets \left\{ (m_{2\smile n}) \in \N^{d-1} \st m \in S\text{ and }m_1 = \mu \right\}$
    \State $r \gets \allbigallsmall(M, k-1)$
    \If{ $\max_{m\in S} m_k > \max_{m\in M} m_k$}
       $r \gets r \cup \{\textsc{out}\}$
    \EndIf
    \If{ $\min_{m\in S} m_k < \min_{m\in M} m_k$ }
       $r \gets r \cup \{\textsc{in}\}$
    \EndIf
    \State \Return $r$
  \EndFunction
\end{algorithmic}
\caption[Computation of the \emph{all large or all small} criterion]{
  Computation of the \emph{all large or all small} criterion
  \begin{description}
    \item[Input] $S\subset \N^{d}$ finite; an integer~$1\leq k \leq d$.
    \item[Output] A subset of~$\left\{ \textsc{in},\textsc{out} \right\}$.
    \item[Specification] Let~$P=\sum_{\ul n\in\N^d} a_{\ul n} z^{\ul n}\in\K[z_{1\smile d}]$ be a polynomial.
      If~$S = \{\ul n \in\N^d\ |\  a_{\ul n} \neq 0 \}$
  then $\allbigallsmall(S,k)$ contains \textsc{out} (resp.
  \textsc{in}) if and only if there exists a nonzero~$\rho\in\overline{\L_d}$ in which~$z_k$ does not appear
  such that~$P(z_1,\dotsc,z_{k-1},\rho,z_{k+1},\dotsc,z_d) = 0$ and~$\mop{lm}(\rho) \succ z_k$
  (resp.~$\mop{lm}(\rho) \prec z_k$).
  \end{description}}
\label{algo:tor}
\end{algo}

\begin{algo}[t]
\centering
\begin{algorithmic}
  \Function{GeomRed}{$F$, $k$}
  \State Decompose~$F$ as~$\sum_{i=1}^r a_i/f_i^{n_i} + P(v)$
  where the~$f_i$'s are irreducible polynomials in~$\K[z_{1\smile d}]$ and $a_{1\smile r}, P\in\K(z_{1\smile k-1},z_{k+1\smile d})[z_k]$.
  \State $S \gets 0$
  \For{$i$ from~$1$ to~$r$}
  \State $\tau \gets \allbigallsmall\left( \{\text{exponents of the monomials of~$f_i$} \}, k \right)$
    \If{$\tau \subset \{ \textsc{in} \}$}
       $S \gets \Res_{z_k =\infty}( a_i/f_i^{n_i} )$
    \EndIf
    \If{$\tau = \left\{ \textsc{in}, \textsc{out} \right\}$}
      \Return \textsc{Fail}
    \EndIf
  \EndFor
  \State \Return $S$
  \EndFunction
\end{algorithmic}
\caption[Geometric reduction]{
  Geometric reduction
  \begin{description}
    \item[Input] $F(z_{1\smile d})$, a rational function; an integer~$1\leq k \leq d$.
    \item[Output] \textsc{Fail} or a rational function $S(z_{1\smile k-1},z_{k+1\smile d})$.
    \item[Specification] 
      If a rational function~$S$ is returned and then~$\res_{z_k} F = S$.
  \end{description}}
\label{algo:ratres}
\end{algo}

\begin{example}
  Continuing Example~\ref{exem:simpl1}, let us consider~$\res_z F_1$,
  where~$F_1=1/(t(1+z)^d+z-t)$, with~$t\prec z$.  The denominator factors
  as~$z g$, where~$g=1+t\sum_{k=0}^{d-1} \binom{d}{k+1}z^k$.  Thus, the leading
  monomial of any root of~$f$ is~$t^{-\frac{1}{d-1}}$, which it not~$\prec z$.
  All the roots of~$g$ are {large}, so only the pole~$z=0$ remains and
  $\res_z F_1 = \Res_{z=0} F_1$, by Proposition~\ref{prop:geomred}.
  Thus~$\res_z F_1 = \frac{1}{1+dt}$ and it follows that
  \[ \sum_{k=0}^n (-1)^k \binom{n}{k}\binom{dk}{n} = (-d)^n. \]
  If we computed the Picard-Fuchs equation of~$F_1$, we would find a
  differential equation of order~$d$ because the Picard-Fuchs operator
  annihilates \emph{all} the periods, and the periods associated to the roots
  of~$g$ are algebraic of degree~$d-1$.
\end{example}

\section{Optimizations}
\label{sec:opt}

\subsection{Infinite sums}
So far, we have only considered binomial sums in which the bounds of the~$\sum$
symbols are finite and explicit.  It is possible and desirable to consider
infinite sums, or more exactly \emph{syntactically} infinite sums which in fact
reduce to finite sums but whose summation bounds are implicit.
This often leads to simpler integral representations.

For example, in the sum~$\sum_{k=0}^n \binom{n}{k}$, the upper summation
bound~$n$ is not really useful when~$n\geq0$, we could as well
write~$\sum_{k=0}^\infty\smash{\binom{n}{k}}$, which defines the same sequence.
It is possible to adapt Algorithm~\ref{algo:sumtoct} to handle infinite sums
as long as the underlying summations in the field of iterated Laurent series are
convergent.  Recall that a geometric sum~$\sum_{n\geq 0}f^n$ converges in the field~$\L_d$
if and only if~$\lm(f)\prec 1$,
see Lemma~\ref{lem:geomsum}.

To compute  an integral representation of binomial sums involving infinite
summations, the principle is to proceed as in~\S\ref{sec:gfbs} except when an 
infinite geometric sum~$\sum_{n\geq 0} f$ in~$\L_d$ shows up, where we check that~$\lm(f)\prec
1$ so that Lemma~\ref{lem:geomsum} applies.  If it does, then the summation is
performed and the computation continues.  If it does not, then the binomial sum is
simply rejected.

\begin{example}
  Consider the binomial sum~$u_n =
  \sum_{k=0}^\infty \binom{n}{k}$, for~$n \geq 0$.
  Note that Algorithm~\ref{algo:sumtoct} applied to~$\sum_{k=0}^n \binom{n}{k}$
  returns
  \[ \sum_{k=0}^n \binom{n}{k} = [1]\left((1+z)^n \frac{z-z^{-n}}{z-1}\right). \]
  We proceed as in~\S\ref{sec:gfbs} except
  that infinite geometric sums in~$\L_d$ are valid when
  Lemma~\ref{lem:geomsum} applies.  Firstly~$\binom{n}{k} = [1] (1+z)^n
  z^{-k}$.  Then we consider the infinite sum $\sum_{k=0}^\infty (1+z)^n z^{-k}$.
  Since~$1/z \succ 1$, it does not converge, so the
  binomial sum is rejected.
  And indeed, when~$n < 0$ the sum~$\sum_{n=0}^\infty\binom{n}{k}$
  has infinitely many non zero terms.
  If we change~$\binom{n}{k}$ into~$\binom{n}{n-k}$, which is the same when~$n \geq 0$,
  we obtain
  \[ u_n = \sum_{k= 0}^\infty [1] (1+z)^n z^{k-n-1} = [1] \sum_{k= 0}^\infty
  \frac{1}{z}\left(1+\frac{1}{z}\right)^n z^{k} = [1]
  \frac{1}{z(1-z)}\left(1+\frac{1}{z}\right)^n, \]
  where this time the sum converges.
\end{example}

\subsection{Building blocks}\label{sec:building-blocks}
Besides the binomial coefficient~$\smash{\binom{n}{k}} = [1] (1+z)^n z^{-k}$,
we have found useful to have additional building blocks
to extend Algorithms~\ref{algo:sumtoct} and~\ref{algo:sumtores}
that contruct integral representations.
Without enlarging the
class of binomial sums, one can add any sequence of the form~$\ul n \mapsto
[1]R_0 R_1^{n_1}\dotsm R_d^{n_d}$, where~$R_{0\smile d}$ are rational
functions, as a new building block. Judicious extra building blocks may give
simpler integral representations or speed up the computation.
For example, when working with Motzkin numbers
\[ M_n \eqdef \sum_{k=0}^\infty (-1)^{n-k} \binom{n}{k} \left( \binom{2k+2}{k+1} - \binom{2k+2}{k+2} \right), \]
one may add the new building block
\[ M_n = [1] \frac{(1-x)(1+x)^2}{x} \left( \frac{1+x+x^2}{x} \right)^n. \]

In the applications below, we made use of alternative definitions of the binomial coefficient and their corresponding representations:
\begin{align}
  \label{eqn:natbinomial} \binom{n}{k}' &\eqdef H_n H_k \binom{n}{k} = [1] \frac{x^{k-n}y^{-k}}{1-x-y},\\
  \binom{n}{k}'' &\eqdef \binom{n}{n-k} = [1]\frac{1}{(1-z)^{k+1} z^{n-k}}.
\end{align}
They differ from the binomial coefficient, as defined in~\S\ref{sec:algbinomsums}, only when~$n < 0$ and agree on their non zero values (see also Figure~\ref{fig:support-binomial}).

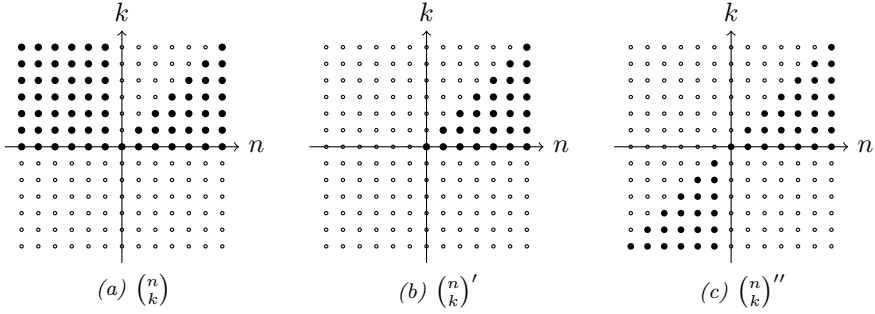
\begin{figure}[t]
  
  \subfloat[$\binom{n}{k}$]{
  \begin{tikzpicture}[scale=.22]
    \draw [->] (-7,0) -- (7, 0) node (naxis) [right] {$n$};
    \draw [->] (0,-7) -- (0,7) node (kaxis) [above] {$k$};

    \foreach \k in {-6,...,6}
      \foreach \n in {-6,...,6} {
        \ifthenelse{\k > -1 \AND \( \cnttest{\k}{<=}{\n} 
      \OR \n < 0 \)}{ \fill (\n,\k) circle[radius = .22]; }{\draw ( \n,\k) circle[radius=.1];} 
      }
    \end{tikzpicture}}
    \quad
  \subfloat[$\binom{n}{k}'$]{
  \begin{tikzpicture}[scale=.22]
    \draw [->,] (-7,0) -- (7, 0) node (naxis) [right] {$n$};
    \draw [->,] (0,-7) -- (0,7) node (kaxis) [above] {$k$};

    \foreach \k in {-6,...,6}
      \foreach \n in {-6,...,6} {
        \ifthenelse{ \cnttest{\k}{<=}{\n}\AND  \k > -1}{ \fill (\n,\k) circle[radius = .22]; }{\draw ( \n,\k) circle[radius=.1];} 
      }
    \end{tikzpicture}}
    \quad
  \subfloat[$\binom{n}{k}''$]{
  \begin{tikzpicture}[scale=.22]
    \draw [->,] (-7,0) -- (7, 0) node (naxis) [right] {$n$};
    \draw [->,] (0,-7) -- (0,7) node (kaxis) [above] {$k$};

    \foreach \k in {-6,...,6}
      \foreach \n in {-6,...,6} {
        \ifthenelse{\cnttest{\k}{<=}{\n} \AND \( \n < 0
      \OR \k > -1 \)}{ \fill (\n,\k) circle[radius = .2]; }{\draw ( \n,\k) circle[radius=.1];} 
      }
    \end{tikzpicture}}

  \caption{Supports of some variants of the binomial coefficient.}
  \label{fig:support-binomial}
\end{figure}

\subsection{Summation over a polyhedron}\label{sec:brion-barvinok}
As in~\S\ref{sec:sumpoly},
let~$\Gamma\subset\R_\geq^{d+e}$ be a rational polyhedron and~$u:\Z^{d+e}\to\K$ be a
binomial sum.  Let us consider the sequence
\[ v : \ul n\in \N^d \mapsto \sum_{\ul m\in \Z^e} u_{\ul n,\ul m} \mathbf1_{\Gamma}(\ul n,\ul m), \]
where~$\mathbf1_{\Gamma}(\ul n,\ul m) = 1$ if~$(\ul n,\ul m)\in\Gamma$ and~$0$ otherwise.
Proposition~\ref{prop:somme-poly} shows that~$v$ is a binomial sum under an additional finiteness hypothesis on~$\Gamma$,
and according to Corollary~\ref{coro:bs-res}, there exists a rational function~$R(t_{1\smile d},z_{1\smile r})$ such that
\[ \sum_{\ul n \in \N^d} v_{\ul n} \ul t^{\ul n} = \res_{z_{1\smile r}} R. \]
It is possible to circumvent the construction used in the proof of Proposition~\ref{prop:somme-poly}
and compute directly such a rational function~$R$ given two ingredients: firstly, the generating function of~$\Gamma$
\[ \varphi_\Gamma(t_{1\smile d},s_{1\smile e}) \eqdef \sum_{(\ul n, \ul m)\in \Z^{d+e}} \mathbf1_{\Gamma}(\ul n,\ul m) \ul t^{\ul n} \ul s^{\ul m}, \]
which is known to be a rational function \parencite[e.g.][]{Bri88};
and secondly, a representation of the binomial sum~$u$ as
\[ u_{\ul n,\ul m} = [1] R T_1^{n_1} \dotsm T_d^{n_d} S_1^{m_1} \dotsm S_e^{m_e} \]
for some rational functions~$R$, $T_{1\smile d}$ and~$S_{1\smile e} \in \K(z_{1\smile r})$.
Then, with~$t_{1\smile e} \prec z_{1\smile d}$,
\begin{align*}
  \sum_{\ul n \in \N^d} v_{\ul n} \ul t^{\ul n}
    &= [1] R \sum_{(\ul n, \ul m)\in \Z^{d+e}} \mathbf1_{\Gamma}(\ul n,\ul m) (t_1T_1)^{n_1} \dotsm (t_d T_d)^{n_d} S_1^{m_1} \dotsm S_e^{m_e}\\
    &= \res_{z_{1\smile r}} \frac{R \cdot \varphi_\Gamma(t_1T_1,\dotsc,t_dT_d,S_1,\dotsc,S_d)}{z_1\dotsm z_d},
\end{align*}
provided that the sums converge in~$\L_{d+e}$.  This method is interesting
because it is known how to compute efficiently compact representations of the
rational function~$\varphi_\Gamma$ \parencite{Barvinok2008}.

\section{Applications}
\label{sec:applications}

\subsection{Andrews-Paule identity}\label{par:binom-exem}

We detail the proof of the following identity:
\begin{equation}
  \sum_{i=0}^n \sum_{j=0}^n\binom{i+j}{j}^2\binom{4n-2i-2j}{2n-2i} = (2n+1)\binom{2n}{n}^2.
  \label{equ:andrews-paule}
\end{equation}
This identity appeared first as a problem in the \emph{American Mathematical Monthly} \parencite{AMME3376}
and was subsequently proved by \textcite{AndrewsPaule1993} and \textcite{Wegschaider1997}
using various tools from the method of creative telescoping.
It attracted attention because the theory of creative telescoping was unable to give a complete automated proof at that time.

Let~$u_n$ denote the left-hand side. It can be written as an infinite sum
\[ u_n = \sum_{i=0}^\infty \sum_{j=0}^\infty \left.\binom{i+j}{j}'\right.^2 \binom{(2n-2i)+(2n-2j)}{2n-2i}',\]
where~$\binom{n}{k}'$ is the \emph{natural} binomial defined by~$\binom{n}{k}'
= H_n H_k \binom{n}{k}$, see~\S\ref{sec:building-blocks}.  Of course, we could
have stuck to the former definition and used finite sums, but while the natural
binomials introduce two variables instead of one in the integral
representation---see Equation~\eqref{eqn:natbinomial}---the integral
representation obtained after the geometric reduction step is often simpler
when using the natural binomial.

So we obtain the following integral representation:
\[ \sum_{n\geq 0} u_n t^n = \res_{z_{1\smile 6}}
  \tfrac{z_1z_2}{(z_1^2z_2^2-t)(z_4z_6-z_1^2)(z_3z_5-z_2^2)(1-z_3-z_4)(1-z_5-z_6)(1-z_1-z_2)}, \]
with the ordering~$t \prec z_1 \prec z_2 \prec z_3 \prec z_4 \prec z_5 \prec z_6$.
As expected, each binomial coefficient brings two extra variables so we end up
with six variables in addition to the parameter~$t$. 

Geometric reduction applies successively with respect to the variables~$z_1$,
$z_3$, $z_4$ and~$z_5$.  For example, the poles w.r.t. the variable~$z_1$,
are~$\{\pm tz_2^{-1/2}\}$,
$\{\pm z_4^{1/2}z_6^{1/2}\}$ and~$\left\{1-z_2\right\}$, gathered by conjugacy classes.
The first pair of conjugate poles is~$\prec z_1$ whereas the second one is~$\succ
z_1$.  The rational root~$1-z_2$ is~$\succ z_1$.  Thus, by application of Proposition~\ref{prop:geomred},
\[ \sum_{n\geq 0} u_n t^n = \res_{z_{2\smile 6}} \frac{(1-z_2) z_2^3}{\left(t-(1-z_2)^2 z_2^2\right) (1-z_3-z_4) \left(z_2^2-z_3 z_5\right) (1-z_5-z_6) \left(t-z_2^2 z_4 z_6\right)}. \]
In the end, repeated application of Algorithm~\ref{algo:ratres} leads to
\[ \sum_{n\geq 0} u_n t^n = \res_{z_{4},z_6}
  \frac{(z_6-1)z_6^3}{(t-z_6^2(z_6-1)^2)((z_4-1)t-z_4z_6^2(z_6^2+z_4-1))}. \]
Using the algorithm of \textcite{Lai15}, we obtain (in about one second) a differential operator annihilating~$\sum_{n\geq 0} u_n t^n$:
\begin{multline*}
  \allowdisplaybreaks[1]
 16t^4(256t^2+736t+81)(16t-1)^2\partial_t^6\\
  +16t^3(16t-1)(86016t^3+256256t^2+20976t-1053)\partial_t^5\\
  +4t^2(36601856t^4+113760256t^3+6103168t^2-908088t+14823)\partial_t^4\\
  +16t(22691840t^4+75716608t^3+6677824t^2-459552t+3645)\partial_t^3\\
  +(305827840t^4+1109626112t^3+139138736t^2-4247073t+9720)\partial_t^2\\
  +(60272640t^3+244005120t^2+42117840t-374625)\partial_t+691200t^2+3369600t+996300
\end{multline*}
The roots of the indicial equation are~$0$, $1$, $-\frac12$ et~$\frac12$.  This
differential operator corresponds to the Picard-Fuchs equation associated to
the integral representation, but it is not the minimal-order operator
annihilating~$\sum_{n\geq 0} u_n t^n$.  This may happen because the
differential operator that we compute annihilates more that simply the residue
in which we ar interested: it annihilates every period of the integral of the
rational function inside the residue.  Of course, this is not a issue as long
as we do obtain a differential equation.

Concerning the right-hand side, we find the integral representation
\[ \sum_{n\geq 0}  (2n+1)\binom{2n}{n}^2 t^n =  \res_{z_1,z_2}
  \frac{t+u_2(u_2-1)u_1(u_1-1)}{(t-u_2(u_2-1)u_1(u_1-1))^2}, \]
which cannot be simplified further with the geometric reduction.
We compute (in about $0.1$s) the annihilating operator: $t(16t-1)\partial_t^2+(48t-1)\partial_t+12$.
The Andrews-Paule identity follows with the equality test described in~\S\ref{sec:diffeqseries}.

In this case the right-hand side is a hypergeometric sequence, so it can be
\emph{discovered} automatically: 
the differential equation of order~6 leads to a recurrence relation of order~4 for~$u_n$
from which the algorithm of \textcite{Pet92} 
finds the hypergeometric solutions and the initial conditions are enough to identify the right-hand side.

\subsection{Several known identities}
\label{sec:known-identities}

This section shows the integral representations and the Picard-Fuchs equations
appearing in the proofs of known identities.  The integral
representations have been obtained with the method presented in this article and
the variants presented in~\S\ref{sec:opt}.
Note that the computation of the annihilating operator never takes longer than~$4$ seconds.

\subsubsection{\textcite{Str94}}\label{sec:strehl}
  $\displaystyle \sum_{k=0}^n \binom{n}{k}^2 \binom{n+k}{k}^2 = \sum_{k=0}^n \binom{n}{k} \binom{n+k}{k} \sum_{j=0}^k \binom{k}{j}^3.$

This identity relates the Apéry numbers (left) and Franel numbers (the inner sum in the right-hand side).

  \medskip
  $\displaystyle  \text{g.f.l.h.s\footnotemark} = \res_{z_{1\smile 3}} \frac{1}{(1-z_1) (1-z_2) (1-z_3)z_1z_2 z_3-t (z_1+z_2 z_3-z_1 z_2 z_3)} $
  \footnotetext{Generating function of the left-hand side} 

  \medskip
  $\displaystyle  \text{g.f.r.h.s\footnotemark}
    = \res_{z_{1\smile 3}} \tfrac{1}{(1-z_1) (1-z_2) (1-z_3)z_1z_2 z_3-t (1-z_3-z_2 (1-(2+z_1 (1-z_2) (1-z_3)) z_3))} $
  \footnotetext{Generating function of the right-hand side} 
  \medskip
  
  $\displaystyle \text{ann. op.\footnotemark} =
    t^2(t^2-34t+1)\partial_t^3+3t(2t^2-51t+1)\partial_t^2+(7t^2-112t+1)\partial_t+t-5$
  \footnotetext{Annihilating operator of both right and left-hand sides}

\subsubsection{\textcites[33]{GraKnuPat89}[\S5.7.6]{Wegschaider1997}}
\[ \sum_{r \geq 0}\sum_{s \geq 0}(-1)^{n+r+s} \binom{n}{r}\binom{n}{s}\binom{n+s}{s}\binom{n+r}{r}\binom{2n-r-s}{n} = \sum_{k\geq 0} \binom{n}{k}^4 \]

  \medskip
  $\displaystyle  \text{g.f.l.h.s} = \res_{z_{1\smile 3}} \frac{1}{(1-z_1) (1-z_2) (1-z_3)z_1z_2 z_3 + (z_2-z_3)(z_1-z_3)t}$

  \medskip
  $\displaystyle  \text{g.f.r.h.s} = \res_{z_{1\smile 3}} 
  \frac{1}{(1-z_1) (1-z_2) (1-z_3)z_1z_2 z_3 +  (1-z_1-z_2) z_3 t  -(1-z_1) (1-z_2) t}$
  \medskip

   $\displaystyle \text{ann. op.} = t^2(4t+1)(16t-1)\partial_t^3+3t(128t^2+18t-1)\partial_t^2t+(444t^2+40t-11)\partial_t+60t+2$

  \subsubsection{Dent} The following identity is due to Dent and used as an example by \textcite[][90]{Wegschaider1997}:
  \[ \displaystyle \sum_{k=0}^{\mathclap{n_1+2n_2}} \sum_{j \geq 0} (-1)^j \binom{k}{j}\binom{2n_2+n_1-k}{2n_2-j}\binom{n_1}{k-j} = 2^{n_1}\binom{n_1+n_2}{n_1}, \text{with $n_1, n_2 \geq 0$}. \]

  \medskip
  $\displaystyle  \text{g.f.l.h.s} = \text{g.f.r.h.s} = \frac{1}{1-2t_1-t_2}$
  \medskip

  Here, the generating series is rational and geometric reduction performs the entire computation, there is no need to compute a Picard-Fuchs equation.

\subsubsection{\textcite{Dix91}} $\displaystyle \sum_{k=0}^{2n} (-1)^k \binom{2n}{k}^3 = (-1)^n \frac{(3n)!}{n!^3}$

  \medskip
  $\displaystyle \text{g.f.l.h.s} = \res_{z_{1\smile 2}} \frac{(1-z_2)(1-z_1)z_1z_2}{z_1^2z_2^2(1-z_2)^2(1-z_1)^2-(1-z_1-z_2)^2t}$
  
  \medskip
  $\displaystyle  \text{g.f.r.h.s} =\res_{z_{1\smile 3}} \frac{1}{t+z_1z_2(1-z_1-z_2)}$

  \medskip
  $\displaystyle \text{ann. op.} = t(27t+1)\partial_t^2+(54t+1)\partial_t+6$

\subsubsection{Moriarty~\parencite[see][11]{Egorychev1984}}
  \[\displaystyle \sum_{k=n}^{m} (-4)^k \binom{k}{m} \frac{n}{n+k}\binom{n+k}{2k} = (-1)^n 4^m \frac{n}{n+m}\binom{n+m}{2m} \]

  Because of the division by~$n+m$, the right-hand side is not obviously a binomial sum.
  However, it becomes obvious after observing that
  \[ \frac{n}{n+k}\binom{n+k}{2k} = \binom{n+k-1}{2k} + \frac12 \binom{n+k-1}{2k-1}. \]

  \medskip
  $\displaystyle  \text{g.f.l.h.s} = \text{g.f.r.h.s} = \frac12\frac{(1-t_1)(1+t_1)}{t_1^2+4t_1t_2+2t_1+1}$

Here again it is a rational power series and the geometric reduction finds it.

\subsubsection{\textcite[Theorem 1.2]{DavEgoKri15}}
  
\begin{multline*}
   1 + \sum_{q=1}^\infty 2^{q-1}\binom{n_2}{q}\left(
    \sum_{m=1}^{n_1/2} \binom{m-1}{q-1} + \sum_{m=1}^{n_1} \binom{m-1}{q-1} \right) \\=
    \sum_{q=1}^\infty 2^{q-1}\binom{n_2}{q} \left( \binom{n_1}{q} + \binom{\lfloor n_1/2\rfloor}{q} \right)
\end{multline*}
 
  $\displaystyle  \text{g.f.l.h.s} = \text{g.f.r.h.s} =
  \frac{t_1t_2(1+t_1-t_2-t_1t_2-2t_1^2-2t_1^2t_2)}{(1-t_2)(1-t_1)(1-t_2-t_1^2t_2-t_1^2)(1-t_1-t_2-t_1t_2)}$
  \medskip

  The summation bound~$n_1/2$ and the integer part~$\lfloor n_1/2 \rfloor$
  may look problematic, until we observe that
  \begin{align*}
  \sum_{m=1}^{n/2} \binom{m-1}{q-1} &= \sum_{m=1}^\infty \binom{m-1}{q-1} H_{n-2m} \\\quad\text{and}\quad
    \binom{\lfloor n/2\rfloor}{q} &= \sum_{k=0}^\infty \binom{k}{q} \left( \delta_{2k-n} + \delta_{2k+1-n} \right),
  \end{align*}
  using the binomial sums~$\delta$ and~$H$ defined in~\S\ref{sec:algbinomsums}.

\subsubsection{\textcite[Theorem 1.1]{DavEgoKri15}}
\begin{multline*}
    2 + \sum_{q=1}^\infty \binom{n_2}{q}\left( \sum_{m=1}^{n_1/2} \binom{m-1}{q-1} + \sum_{k_1=1}^\infty \dotsb \sum_{k_q=1}^\infty \delta_{k_1+\dotsb+k_q-m}\prod_{i=1}^q (k_i+1) \right) \\
    = \binom{n_1+2n_2}{2n_2} + \binom{\lfloor n_1/2\rfloor + n_2}{n_2}
\end{multline*}

  $\displaystyle  \text{g.f.l.h.s} = \text{g.f.r.h.s} =
  \frac{t_1t_2(1-t_2-2t_1^2 +t_1^3-2t_1^2)}{(1-t_2)(1-t_1)(1-t_2-t_1^2)(1-t_2-2t_1+t_1^2)}$
  \medskip

  Again, the left-hand side does not have the appearance of a binomial sum until we remark that
  \[ \sum_{k_1=1}^\infty \dotsb \sum_{k_q=1}^\infty \delta_{k_1+\dotsb+k_q-m}\prod_{i=1}^q (k_i+1) = [1]\left(
  \left(\frac{t(2-t)}{(1-t)^2} \right)^q \frac{1}{t^m} \right). \]

\subsubsection{\textcite{AMM11914}}
\[ \sum_{k=1}^n (-4)^{-k}\binom{n-k}{k-1} \sum_{j=1}^{3m} (-2)^{-j} \binom{n+1-2k}{j-1} \binom{m-k}{3m-j} = 0, \quad n,m > 0, \]
where we consider the binomial coefficient as defined in~\S\ref{sec:def} and
not one of the variants of~\S\ref{sec:building-blocks}.  This is an important
clarification because negative values may appear in the upper arguments of the
binomial coefficients.  The geometric reduction is enough to prove this
equality, no integration step is required.

%\subsubsection{\textcite{AMM11916}}
%\[ \binom{n+r}{r} \sum_{k=0}^{s-1} \binom{r+k}{r-1}\binom{n+k}{k}
%= \binom{n+s}{s} \sum_{k=0}^{r-1} \binom{s+k}{s-1}\binom{n+k}{k} \]
%$\displaystyle \text{g.f.l.h.s} = \text{g.f.l.r.h.s} =

\subsection{Proof of some conjectures}
\label{sec:proof-conjectures}
\subsubsection{Le Borgne's identity}

The following identity for Baxter's numbers arises as a conjecture in an
unpublished work by Yvan Le~Borgne. With the methods presented here we can
prove it automatically.
\begin{multline*}
  1 + F_n^{-1, -1}+2F_n^{0, 0}-F_n^{0, 1}+F_n^{1, 0}-3F_n^{1, 1}+F_n^{1, 2}-F_n^{3, 1}+3F_n^{3, 2}\\-F_n^{3, 3}-2F_n^{4, 2}+F_n^{4, 3}-F_n^{5, 2}
  = \sum_{m=0}^{n} \frac{\binom{n+2}{m}\binom{n+2}{m+1}\binom{n+2}{m+2}}{\binom{n+2}{1}\binom{n+2}{2}}, \\
  \text{where } F_n^{a,b} = \sum _{d=0}^{n-1} \sum_{c=0}^{d-a} \tbinom{d-a-c}{c}\tbinom{n}{d-a-c}\left(\tbinom{n+d+1-2a-2c+2b}{n-a-c+b}-\tbinom{n+d+1-2a-2c+2b}{n+1-a-c+b}\right).
\end{multline*}

The sum in the right-hand side does not have the appearance of a binomial sum (even if it is such), but the identity becomes clearly an identity between binomial sums after multiplication by~$\binom{n+2}{1}\binom{n+2}{2}$.
Alternatively, one can use the following equality to write the right-hand side as a binomial sum:
\[ \frac{\binom{n+2}{m}\binom{n+2}{m+1}\binom{n+2}{m+2}}{\binom{n+2}{1}\binom{n+2}{2}} = \det
  \begin{pmatrix}
    \binom{n}{k} & \binom{n}{k+1} & \binom{n}{k+2} \\
    \binom{n}{k-1}& \binom{n}{k}& \binom{n}{k+1} \\
    \binom{n}{k-2} & \binom{n}{k-1} & \binom{n}{k}
  \end{pmatrix},
\]

\subsubsection{Identities from \textcite{BreOhtOsb14}} We have been able to
prove the following identities conjectured by \textcite{BreOhtOsb14}.  The
left-hand sides involve absolute values of nonlinear polynomials.  They are
nevertheless binomial sums for two reasons.  The first one is that all the
nonlinear polynomials under consideration split into linear or positive
factors.  For example~$|i^3-j^3| = |i-j|(i^2+ij+j^2)$ and~$|i-j| =
(i-j)(H_{i-j}-H_{j-i})$ is a binomial sum.

The second reason, which we used for the computation, is that we can eliminate
the absolute values by using the symmetries of the sums. For example
\[ \sum_{i,j} \binom{2n}{n+i}\binom{2n}{n+j}|i^3-j^3| = 4\sum_{i=0}^n\sum_{j=-i}^{i-1} \binom{2n}{n+i}\binom{2n}{n+j}i^3. \]

The integral representations obtained are too lengthy to be presented
here, they can be found online\footnote{\url{https://github.com/lairez/binomsums}}.  Since the
right-hand sides are hypergeometric sequences, they can be computed from the
recurrence relations satisfied by the left-hand sides with the algorithm of \textcite{Pet92}.

\[\sum_{i,j} \binom{2n}{n+i}\binom{2n}{n+j}|i^3-j^3| = \frac{2n^2(5n-2)}{4n-1}\binom{4n}{2n}
  \tag*{\parencite[Eq.~5.7]{BreOhtOsb14}} \]

\[\sum_{i,j} \binom{2n}{n+i}\binom{2n}{n+j}|i^5-j^5| =
  \frac{2n^2(43n^3-70n^2+36n-6)}{(4n-1)(4n-3)}\binom{4n}{2n}
  \tag*{\parencite[Eq.~5.8]{BreOhtOsb14}} \]

\[\sum_{i,j} \binom{2n}{n+i}\binom{2n}{n+j}|i^7-j^7| =
  \tfrac{2n^2(531n^5-1960n^4+2800n^3-1952n^2+668n-90)}{(4n-1)(4n-3)(4n-5)}\binom{4n}{2n}
  \tag*{\parencite[Eq.~5.9]{BreOhtOsb14}} \]

\[\sum_{i,j} \binom{2n}{n+i}\binom{2n}{n+j}|ij(i^2-j^2)| =
  \frac{2n^3(n-1)}{2n-1}\binom{2n}{n}^2
  \tag*{\parencite[Eq.~5.12]{BreOhtOsb14}} \]

\[\sum_{i,j} \binom{2n}{n+i}\binom{2n}{n+j}|i^3j^3(i^2-j^2)| =
  \frac{2n^4(n-1)(3n^2-6n+2)}{(2n-3)(2n-1)}\binom{2n}{n}^2
  \tag*{\parencite[Eq.~5.14]{BreOhtOsb14}} \]

\subsection{Computational limitations}
The method is mainly limited by the integration step.  When the integral
representation has more than four variables (in addition to the parameter and after the geometric reduction),
then computation of the Picard-Fuchs equation becomes challenging.  For example,
in Strehl's second identity
\[ \sum_{k=0}^n \binom{n}{k}^3\binom{n+k}{k}^3 = \sum_{k=0}^\infty \binom{n}{k}\binom{n+k}{k}  \sum_{i=0}^\infty \binom{k}{i}^2\binom{2i}{i}^2 \binom{2i}{k-i}, \]
the integral representation of the generating function of each side has five
variables and a parameter.  For the left-hand side, we obtain
\[ \res_{z_{1\smile 5}} \tfrac{1}{(z_1z_3z_4z_5+z_2z_3z_4z_5 -z_1z_2z_3z_4z_5-z_3z_4z_5+z_1z_2)t+z_1z_2z_3z_4z_5(1-z_1)(1-z_2)(1-z_3)(1-z_4)(1-z_5)}. \]
The integral representation of the right-hand side is more complicated
ant still has five variables, in addition to the parameter. 
The computation  requires several hours with the current algorithms.
Without the
geometric reduction, it involves nine variables and a parameter.

\raggedright
\printbibliography

\end{document}